\documentclass[a4paper,11pt]{article}
\usepackage[a4paper]{geometry}
\usepackage[english]{babel}
\title{The higher spin Laplace operator in several vector variables}
\author{David Eelbode$^{\ast\ast}$, Tim Raeymaekers$^\ast$ and Matthias Roels$^{\ast\ast}$} 
\date{\small{$^\ast$ Clifford Research Group, Ghent University, Galglaan 2, 9000 Gent, Belgium\\
$^{\ast\ast}$ University of Antwerp, Middelheimlaan 2, 2020 Antwerpen, Belgium}}

\usepackage{graphicx}
\usepackage{epsfig}
\usepackage{subfigure}
\usepackage{float}
\restylefloat{figure}
\restylefloat{table}
\usepackage[font=small,format=plain,labelfont=bf,up,textfont=it,up]{caption}

\usepackage{enumitem}

\usepackage[vcentermath]{youngtab}

\usepackage{amsmath,amsfonts,amsthm,amssymb,mathdots}
\usepackage{mathrsfs}
\usepackage{mathtools}
\usepackage{accents}
\usepackage{dsfont}
\usepackage{cases}

\usepackage[english]{babel}

\usepackage{tikz}
\usetikzlibrary{arrows,decorations.markings,positioning,fit,calc, matrix}
\usetikzlibrary{trees, shapes, fit, calc}
\usetikzlibrary{graphs}
\usetikzlibrary{backgrounds}

\def\DynkinNodeSize{4.5pt}
\def\DynkinArrowLength{5pt}
\tikzset{
  dnode/.style={
    circle,
    inner sep=0pt,
    minimum size=\DynkinNodeSize,
    fill=white,
    draw},
  dbnode/.style={
    circle,
    inner sep=0pt,
    minimum size=\DynkinNodeSize,
    fill=black,
    draw},
  middlearrow/.style={
    decoration={markings,
      mark=at position 0.6 with
      {\draw (0:0pt) -- +(+135:\DynkinArrowLength); \draw (0:0pt) -- +(-135:\DynkinArrowLength);},
    },
    postaction={decorate}
  },
  sedge/.style={
  },
  dedge/.style={
    middlearrow,
    double distance=1.4pt,
  },
  tedge/.style={
   middlearrow,
   double distance=2.83pt+\pgflinewidth,
   postaction={draw}, 
   },
  circle dotted/.style={
  	shorten >=2pt,
  	shorten <=4pt,
  	dash pattern=on 0pt off 8pt,
  	line cap=round,
  	line width=1pt
  },
}

\newcommand{\D}{\partial} 
\newcommand{\mE}{\mathbb{E}} 

\newcommand{\mcC}{\mathcal{C}}
\newcommand{\mcD}{\mathcal{D}}

\newcommand{\mcH}{\mathcal{H}}

\newcommand{\mcJ}{\mathcal{J}}

\newcommand{\mcN}{\mathcal{N}}
\newcommand{\mcO}{\mathcal{O}}
\newcommand{\mcP}{\mathcal{P}}

\newcommand{\mcU}{\mathcal{U}}

\newcommand{\mV}{\mathbb{V}}

\newcommand{\g}{\mathfrak{g}}
\newcommand{\h}{\mathfrak{h}}

\newcommand{\gok}{\mathfrak{k}}

\newcommand{\s}{\mathfrak{s}}
\newcommand{\gt}{\mathfrak{t}}

\newcommand{\gl}{\mathfrak{gl}}

\newcommand{\so}{\mathfrak{so}}
\newcommand{\sym}{\mathfrak{sp}}
\newcommand{\spl}{\mathfrak{sl}}

\newcommand{\SO}{\operatorname{SO}}
\newcommand{\Pin}{\operatorname{Pin}}

\newcommand{\Aut}{\operatorname{Aut}}

\newcommand{\Span}{\operatorname{Span}}
\newcommand{\Alg}{\operatorname{Alg}}

\newcommand{\quotient}[2]{\left.\raisebox{.2em}{$#1$}\middle/\raisebox{-.2em}{$#2$}\right.}

\newcommand{\mR}{\mathbb{R}}
\newcommand{\mC}{\mathbb{C}}
\newcommand{\mN}{\mathbb{N}}
\newcommand{\mZ}{\mathbb{Z}}

\newcommand{\inner}[1]{\langle #1\rangle} 
\newcommand{\comm}[1]{\lbrack #1 \rbrack} 
\newcommand{\set}[1]{\left\{ #1 \right\} } 
\newcommand{\norm}[1]{\left| #1 \right|} 
\newcommand{\brac}[1]{\left( #1 \right)} 

\newcommand{\pha}{\phantom{a}}

\newtheorem{definition}{Definition}[section]
\newtheorem{proposition}{Proposition}[section]
\newtheorem{theorem}{Theorem}[section]
\newtheorem{lemma}{Lemma}[section]
\newtheorem{remark}{Remark}[section]
\newtheorem{corollary}{Corollary}[section]
\newtheorem{example}{Example}[section]

\begin{document}
\maketitle

\begin{abstract}
In this paper, an explicit expression is obtained for the conformally invariant higher spin Laplace operator $\mcD_\lambda$, which acts on functions taking values in an arbitrary (finite-dimensional) irreducible representation for the orthogonal group with integer valued highest weight. Once an explicit expression is obtained, a special kind of (polynomial) solutions of this operator is determined.
\end{abstract}

\section{Introduction} 

This paper fits into the current program of constructing conformally invariant operators for higher spin fields, which leads to (massless) equations generalising the 
Maxwell and Dirac equation to arbitrary dimension, and studying their properties and solutions. Over the last decade it has become clear that a better 
understanding of higher spin gauge theories, as pioneered by Vasiliev and Fradkin in \cite{FV1, FV2}, involves the use of higher spin fields of so-called mixed 
symmetry type (i.e. tensor fields described by a Young symmetriser which is neither symmetric or anti-symmetric). In particular we would like to refer to the 
appearance of mixed symmetry fields in the setting of higher spin conformal field theories, see e.g. \cite{CHPT, CPPR, SD}, which has its roots in the conformal 
bootstrap method \cite{FGG}. The approach followed in these papers lies close to the philosophy of the present paper, in which models for arbitrary (finite-
dimensional) irreducible representations involving polynomials satisfying certain systems of equations are used (see also \cite{CSVL, GM}). This is an alternative for 
the well-known abstract index notation for tensor fields, which lends itself nicely to explicit calculations. Using these models, a plethora of results for higher spin 
fields in arbitrary dimension has been obtained within the framework of Clifford analysis, a higher dimensional version of complex analysis which allows one to 
focus on the function theoretical aspects of the theory of differential operators and their solution spaces. Whereas this theory originally focused most of its 
attention on the Dirac operator, see for instance \cite{DSS, GM}, the area of interest has substantially grown once it became clear that also higher spin Dirac 
operators and super versions thereof were elegantly described using similar techniques. We refer for instance to \cite{BSSVL1, BSSVL2, CDB, ER1, ERVdJ} for the 
case of first-order operators. As for the case of second-order operators, for which we refer to e.g. \cite{Br, HM}, we have obtained partial results in \cite{DER1, 
DER2, ER}: in the former paper we focused on the actual construction and function theoretical properties such as fundamental and polynomial solutions for the 
operator acting on symmetric tensor fields, in the latter we developed an algebraic framework for studying for instance reproducing kernels for irreducible 
representation spaces (connected to the projectors calculated in \cite{CHPT}).
\par
The aim of the present paper is to combine these ideas, and apply them to the construction of the most general second-order conformally invariant operator using 
techniques from Clifford analysis (i.e. heavily relying on the polynomial model for higher spin fields). A special case was considered in \cite{DER1} where for 
instance the linearised Einstein equations \cite{Mayor} where written as 
\begin{equation*}
\mcD_2 = \Delta_x-\frac{4}{m+2}\brac{\inner{u,\D_x}-\frac{\norm{u}^2}{m}\inner{\D_u,\D_x}}\inner{\D_u,\D_x}.
\end{equation*}
The subscript two hereby refers to the highest weight $\lambda=(2,0,\ldots,0)$. In this paper, we will start with an arbitrary integer highest weight $\lambda$ and 
the associated operator will be denoted by $\mcD_{\lambda}$. 
The existence of these operators follows from a general classification result (see e.g. \cite{BEastwood, BGG, Lep, Slovak}) or from Branson's paper \cite{Br} in the 
curved setting (see also section 2). Nowadays these operators are usually constructed in terms of the tractor calculus formalism, but this tends to get rather 
complicated due to the presence of all the curvature corrections (see e.g. \cite{Gover}). In this paper we will focus our attention to the flat case. We explain how 
Branson's method, based on generalised gradients \cite{SW}, can be formulated on the level of a transvector algebra, as this allows us to also obtain results on the 
space of solutions. Despite being a partial result only, the solutions obtained in this paper (the so-called solutions of type A) are important for two reasons: first of 
all, our experience from the first-order case has taught us that the `missing' solutions can be obtained from the ones described in the paper using twistor 
operators (see \cite{ERVdJ}). Secondly, the solutions of type A are exactly the solutions satisfying the so-called tranversality conditions which arise in the framework 
of the celebrated ambient method (the mathematical framework underlying both the AdS/CFT correspondence and the tractor calculus). They can thus be seen as 
solutions satisfying a gauge condition, for which we refer to the last section. 
\par
This paper is organised as follows. In Section 2, we use Branson's result to introduce second order conformally invariant operators. We explain how the theories of 
transvector algebras and generalised gradients can be used to construct these second order conformally invariant operators in Section 3 and construct these 
operators explicitly in Section 4. In Section 5, a special type of solutions of such operators is considered, together with the relation of this space of solutions to 
transvector algebras in Section 6.


\section{Branson's result on second-order conformally invariant operators}

In this section, we give a brief overview of the main results of \cite{Br}. Suppose $M$ is an oriented Riemannian manifold of dimension $m>3$ and let $F(M)$ be the frame bundle of $TM$. Since $M$ is a Riemannian 
manifold, it is possible to consider the reduction of $\operatorname{GL}(m)$ to $\SO(m)$ to obtain the orthonormal frame bundle 
$F_o(M)$. Suppose $V_{\lambda}$ is a finite-dimensional irreducible representation of $\SO(m)$ with highest weight
$\lambda=(\lambda_1,\ldots,\lambda_n)$. We may then form the associated vector bundle 
$E_{\lambda}=F_o(M)\times_{\SO(m)}V_{\lambda}$. Every bundle constructed in this way is isomorphic to some tensor bundle 
$E(\lambda)$ obtained by taking tensor products of $TM$ and $T^{\ast}M$. For example, the (co)tangent bundle is associated 
with the standard representation of $\SO(m)$, i.e.  the representation with highest weight $\lambda=(1,0,\cdots,0)$. 
Since $M$ is a Riemannian manifold, we can use the Levi-Civita connection on $E(\lambda)$, which is a map 
$\nabla: \Gamma(E(\lambda))\longrightarrow \Gamma(E(\lambda)\otimes T^{\ast}M)$.
The bundle $E(\lambda)\otimes T^{\ast}M$ is no longer irreducible and can be decomposed into irreducible bundles, i.e. 
\begin{equation*}
E(\lambda)\otimes T^{\ast}M\cong_{\SO(m)}E(\mu_1)\oplus\ldots \oplus E(\mu_{\ell}),
\end{equation*}
where $\mu_j$ is the highest weight of the irreducible bundle $E(\mu_j)$.
A generalised gradient $G_{\lambda,\mu_j}=\operatorname{Proj}_{\mu_j}\circ \nabla$ is then defined as the composition of $\nabla$ 
with the projection on the bundle $E(\mu_j)$. In \cite{Fegan}, Fegan proved that any generalised gradient is a first order conformally invariant 
differential operator and the every conformally invariant first order operator is in fact a generalised gradient. 
We now have gathered the necessary ingredients to formulate the main result \cite{Br}:
\begin{theorem}\label{theorem_Branson}
Given an irreducible vector bundle $E(\lambda)$, there exists a unique (up to a constant multiple) second-order 
conformally invariant operator $D_{\lambda}:E(\lambda)\longrightarrow E(\lambda)$ when $m$ is odd or when $m$ is even and 
$\lambda_n=0$ is zero. When $m$ is odd or when $m$ is even and $\lambda_{n-1}=0$, the operator $D_{\lambda}$
 is explicitly given by
\begin{equation*}
D_{\lambda}:=\sum_{j=1}^{\ell} \brac{\frac{1}{2}+\inner{\lambda+\mu_j+2\rho,\lambda-\mu_j}}^{-1} G_{\lambda, \mu_j}^{\ast}
G_{\lambda, \mu_j}-\frac{R}{2(m-1)},
\end{equation*}
where $R$ is the scalar curvature of $M$, $\inner{\cdot,\cdot}$ is the Killing form of $\so(m)$ and $\rho$ is the sum of the 
fundamental weights of $\so(m)$. If on the other hand $m$ is even and $\lambda_{n-1}\neq0$, then the operator is given by 
$D_{\lambda}:=G_{\lambda, \lambda\pm\varepsilon_{n}}^{\ast}G_{\lambda, \lambda\pm \varepsilon_{n}}$.
\end{theorem}
In the next section, we will construct the operators $G_{\lambda, \mu_j}$ and $G_{\lambda, \mu_j}^{\ast}$ on $\mR^m$. We will then use this theorem to derive an explicit form of $D_{\lambda}$ containing some unknown constants (different from the ones in this theorem) that can be fixed by proving conformal invariance, which is similar to what was done in \cite{DER1}.


\section{Transvector algebras and generalised gradients} 

In this section, we will show that the construction of conformally invariant first-order operators, known as twistor operators (and their duals), on $\mR^m$ can be 
performed in an abstract framework involving transvector algebras. These algebras were used in a recent paper \cite{DER2} to obtain an alternative for the classical 
Howe dual pair $\SO(m) \times \sym(4)$ underlying analysis for matrix variables in $\mR^{2m}$, and have led to explicit projection operators on the irreducible 
summands for certain tensor products. Because the construction of first-order conformally invariant operators is based on the Stein-Weiss method
(see \cite{Fegan, SW}), which essentially amounts to projecting on irreducible summands, this transvector algebra turns out to be particularly useful. 
\par 
In order to formally introduce the notion of a transvector algebra, we start from a decomposition 
$\g = \s \oplus \gt$ for a simple Lie algebra $\g$, where $\s \subset \g$ is a subalgebra and where the subspace $\gt$ defines a 
module for the adjoint (commutator) action of $\s$ (i.e. $[\s,\gt] \subset \gt$). 
Defining $\s^+ \subset \s$ as the space containing the positive root vectors of $\s$, one can then introduce the left-sided ideal 
$J' := \mcU'(\g)\s^+$ in $\mcU'(\g)$. The prime hereby stands for the localisation of the universal enveloping algebra $\mcU(\g)$ 
with respect to the subalgebra $\mcU(\h)$, with $\h$ the Cartan algebra in $\g$ (roughly speaking, this allows to write down fractions containing Cartan elements). 
Defining the normaliser $\operatorname{Norm}(J') := \big\{u \in \mcU'(\g) : J'u \subset J'\big\}$, we finally arrive at the following 
(for more details, we refer the reader to \cite{Molev, Zh}): 
\begin{definition}
For a simple Lie algebra $\g = \s \oplus \gt$, one can define the transvector algebra $Z\big(\g,\s\big) := \operatorname{Norm}(J')/J'$, with $J' = \mcU'(\g)\s^+$. 
\end{definition}
\par
This algebra was first introduced by Zhelobenko in \cite{Zh}, in the framework of so-called extremal systems, such as the Dirac and Maxwell equations. In the works of Zhelobenko and Molev \cite{Molev, Zh}, it was shown that the algebra $Z\big(\g,\s\big)$ is generated by the elements $\pi_\s[u]$, with $u \in \gt$ (the complementary $\s$-module) and $\pi_\s \in \mcU'(\gok)$ the extremal projection operator for the Lie algebra $\s$.
\begin{definition}
The extremal projection operator $\pi_\s$ for a semisimple Lie algebra $\s$ is defined as the unique formal operator, contained in 
an extension of $\mcU'(\s)$ to some algebra of formal series, which satisfies the equations 
$\s^+\pi_\s = 0 = \pi_\s\s^-$ and $\pi_\s^2 = \pi_\s$, with $\s^+$ and $\s^-$ the subspaces
containing positive and negative root vectors respectively. 
\end{definition}
\begin{remark}
Note that the localisation with respect to $\mcU(\h)$ is necessary in order to obtain a proper projection operator (i.e. satisfying the requirement $\pi_\s^2 = \pi_\s$). 
\end{remark}
Transvector algebras are examples of so-called quadratic algebras, which satisfy commutation relations of the form
\begin{equation*}
 [X_a,X_b] = \sum_{c,d}\alpha_{cd}X_cX_d + \sum_{p}\beta_pX_p + \gamma,
\end{equation*}
where $X \in Z(\g,\gok)$ is a generator and where the `constants' $\alpha_{cd}, \beta_p$ and $\gamma$ belong to $\mcU'(\h)$. This thus generalises the classical 
relations for Lie algebras. 

\subsection{Two vector variables}

In \cite{DER2}, we have explicitly determined the commutation relations for $Z(\sym(4),\so(4))$ and used them to 
obtain an alternative for the classical Howe dual pair $\SO(m) \times \sym(4)$. A model for this algebra was obtained using a polynomial realisation for
$\sym(4) = \s \oplus \gt$, given by
\begin{equation*}
\sym(4) = \big(\spl_x(2) \oplus \spl_u(2)\big) \oplus \Span_\mC\big(\inner{u,x}, \inner{\D_u,\D_x}, \inner{x,\D_u}, \inner{u,\D_x}\big),
\end{equation*}
where the notation $\inner{\cdot,\cdot}$ hereby refers to the Euclidean inner product on $\mR^m$ and an operator such as for example $\inner{u,\D_x}$ stands for $\sum_{j=1}^mu_j\D_{x_j}$. 
The subspace $\gt$ is a module for $\s = \so(4) \cong \spl_x(2) \oplus \spl_u(2)$, with 
$\spl_x(2) = \Span(|x|^2,\Delta_x,\mE_x + \frac{m}{2})$ and similarly for $u$. Here, $|x|^2$ is the squared norm of the vector $x$, $\Delta_x$ is the corresponding Laplace operator, and $\mE_x=\langle x, \partial_x \rangle$ is the Euler operator. Acting on these $4$ operators in $\gt$ with the extremal projection operator 
$\pi_{\so(4)}$ leads to a quadratic algebra generated by $4$ elements, which we called $(C,A,S_x,S_u)$ in \cite{DER2}. This algebra 
forms the inspiration for the rest of this section, in which we will show that transvector algebras of the form $Z(\g,\s)$, with 
$\g = \sym(2n)$ a symplectic Lie algebra and $\s$ a subalgebra thereof, are intimitely connected to the classical Stein-Weiss 
method for constructing invariant operators.
The starting point is the following identification, valid for a field
\begin{equation*}
 f_k(x,u) = \sum_{\alpha = 1}^{d_k}\varphi_\alpha(x)H_k^\alpha(u)
\end{equation*}
taking values in the space of traceless symmetric tensors of rank $k$ (for which the space of $k$-homogeneous harmonic 
polynomials provides a model, hence the dummy $u \in \mR^m$ and the subscript $k$ attached to our field): 
\begin{equation*}
\nabla f_k \in \Gamma(\mcH_k\otimes T^{\ast}M)\ \stackrel{1:1}{\longleftrightarrow}\ \sum_{j=1}^m\sum_{\alpha = 1}^{d_k}\frac{\D \varphi_\alpha(x)}{\D x_j}v_jH_k^\alpha(u) \in \mcC^\infty(\mR^m,\mcH_k \otimes \mcH_1),
\end{equation*}
where $\alpha \in \{1,\ldots,d_k\}$ is an index labeling a basis for the $d_k$-dimensional space $\mcH_k(\mR^m,\mC)$. 
\begin{remark}
For a general manifold $M$ one has to see $\set{H_k^\alpha(u)}_{\alpha=1}^{d_k}$ as a local basis of the space of sections of an appropriate bundle. In this paper however, we are working with a trivial bundle.
\end{remark}
We have thus introduced a dummy variable $v \in \mR^m$ to denote the representation according to which the cotangent bundle 
decomposes (as we are working with the standard Euclidean metric, the identification with the tangent bundle is immediate). As is 
well-known, this tensor product $\mcH_k\otimes \mcH_1$ decomposes as a module for $\SO(m)$ and the projection 
on each of these summands provides us with a generalised gradient. This tensor product can now be studied within the framework 
of Fischer decompositions for harmonic polynomials in two vector variables, as $\Delta_u H_k(u) = 0 = \Delta_v v_j$, and this is 
precisely the setting in which we have introduced the algebra $Z(\sym(4),\so(4))$.
For instance, for fixed $\alpha$ and $j$ one has that 
\begin{equation}\label{SW_Dk}
v_jH^\alpha_k(u) = H_{k,1}(u,v) + S_v H_{k+1}(u) + C H_{k-1}(u),
\end{equation}
where $Z(\sym(4),\so(4)) = \Alg(S_u,S_v,A,C)$ in the variables $(u,v) \in \mR^{2m}$ and where $H_{k\pm 1}(u)\in\mcH_{k\pm 1}$ 
and $H_{k,1}(u,v)\in \mcH_{k,1}$. Note that the polynomials appearing at the right-hand side should be indexed by two 
labels $(\alpha,j)$, but we have decided to suppress this notation. As a matter of fact, to get an explicit expression 
representing $\nabla f_k(x,u)$ one has to reintroduce the summation over $\alpha$ and $j$, which then leads to
\begin{equation*}
\nabla f_k(x,u) = f_{k,1}(x,u,v) + S_v f_{k+1}(x,u) + C f_{k-1}(x,u),
\end{equation*}
where the subscripts refer to the values of the fields under consideration. We can now invoke the knowledge obtained in 
\cite{DER2} to construct explicit projection operators. First of all, the action of $A = \pi_{\s}[\inner{\D_u,\D_v}]$ on 
the left-hand side gives 
\begin{equation*}
\sum_{j = 1}^m\sum_{\alpha = 1}^{d_k}\frac{\partial \varphi_\alpha(x)}{\partial x_j}A\big(v_jH^\alpha_k(u)\big) = \pi_{u}[\inner{\D_u,\D_x}] f_k(x,u),
\end{equation*}
Note that the action of $\pi_{\s}$ reduces to $\pi_{u}$ (the extremal projection operator for $\spl(2)_u$) as the projection for 
$v \in \mR^m$ is trivial. 
Since the field above takes its values in the space $\mcH_{k-1}(\mR^m,\mC)$ and the decomposition (\ref{SW_Dk}) is multiplicity-
free, the only non-trivial contribution coming from the action of $A$ on the right-hand side is given by
\begin{equation*}
AC f_{k-1}(x,u) = \left(\frac{H_v - H_u}{H_u + 1} - (H_u + H_v)\right)f_{k-1}(x,u),
\end{equation*}
where $-2H_u = 2\mE_u + m$ (and similarly for $v$). We thus get
\begin{equation*}
AC f_{k-1}(x,u) = \frac{(k+m-3)(2k+m-2)}{2k+m-4}f_{k-1}(x,u).
\end{equation*}
This means that the projection on $\Pi_{k-1} : \mcH_k\otimes \mcH_1 \rightarrow \mcH_{k-1}$ is given by
\begin{equation*}
\Pi_{k-1}\big(\nabla f_k(x,u)\big) = \frac{2k+m-4}{(k+m-3)(2k+m-2)}\inner{\D_u,\D_x} f_k(x,u),
\end{equation*}
Note that this is the expression {\em without} the embedding factor, i.e. the invariant map embedding $\mcH_{k-1} \hookrightarrow \mcH_k \otimes T^{\ast}M$ (it suffices to let the operator $C$ act to include it). Next, something similar can be done for the action of the operator $S_u$ on expression (\ref{SW_Dk}). For the left-hand side, we get
\begin{equation*}
\sum_{j = 1}^m\sum_{\alpha = 1}^{d_k}\frac{\partial \varphi_\alpha(x)}{\partial x_j}S_u\big(v_jH^\alpha_k(u)\big) = \pi_{u}
[\inner{u,\D_x}] f_k(x,u), 
\end{equation*}
a field taking its values in $\mcH_{k+1}(\mR^m,\mC)$. 
\begin{remark}
Note that the action of the operator $\pi_{u}[\inner{u,\D_x}]$ on a field $f_k(x,u)$ can either be read as first acting with 
$\inner{u,\D_x}$ on $f_k(x,u)$ and then projecting on the harmonic part in $u$ (the first action gives a polynomial of degree $(k+1)$ 
in $u$, for fixed $x \in \mR^m$), or as the action of the operator
\begin{equation*}
\pi_{u}[\inner{u,\D_x}] = \inner{u,\D_x}-\frac{1}{2\mE_u+m-4}\norm{u}^2\inner{\D_u,\D_x},
\end{equation*}
which has the projection built into its explicit form. Because we always assume that our fields are harmonic in $u \in \mR^m$, these 
actions are equivalent. 
\end{remark}
Again invoking the algebraic relations in $Z(\sym(4),\so(4))$, one also has that
\begin{equation*}
S_uS_v f_{k+1}(x,u) = (k+1)f_{k+1}(x,u)\ ,
\end{equation*}
which eventually leads to 
\begin{equation*}
\Pi_{k+1}\big(\nabla f_k(x,u)\big) = \frac{1}{k+1}\pi_{\spl(2)}[\inner{u,\D_x}] f_k(x,u).
\end{equation*}
Finally, a simple subtraction leads to the last projection: 
\begin{equation*}
\Pi_{k,1}\big(\nabla f(x,u)\big) = \bigg(1 - S_v\Pi_{k+1} - C\Pi_{k-1}\bigg)\inner{v,\D_x} f_k(x,u),
\end{equation*}
Once these generalised gradients are known, one can determine their formal adjoints. To do so, we work in the space 
$\mcC^\infty_c(\mR^m,\mcH_k) = \mcC^\infty_c(\mR^m,\mC) \otimes \mcH_k$, containing smooth compactly supported functions 
taking values in the space 
$\mcH_k$, such that we can make sense of the defining relation $\inner{G_\lambda f,g} = \inner{ f,G_\lambda^{\ast}g}$, where 
$G_\lambda$ denotes a generalised gradient for the highest weight $\lambda$ and $G_\lambda^*$ its dual (note that the compactly 
supported function $f$ lives in the target space for $G_\lambda$, which obviously depends on the highest weight $\lambda$). 
The inner product is hereby defined (e.g. for basis elements) by means of
\begin{equation*}
\inner{\varphi_\alpha(x) \otimes H^\alpha(u),\psi_\beta(x) \otimes H^\beta(u)} := \int_{\mR^m}\varphi_{\alpha}(x)\overline{\psi_{\beta}(x)}\mathrm{d}x \otimes [H^\alpha(u),H^\beta(u)]_F.
\end{equation*}
Here $[\cdot,\cdot]_F$ denotes the Fischer inner product on the space of values (simplicial harmonics in general), which is defined 
as
\begin{equation*}
[H^{\alpha}(u),H^{\beta}(u)]_F:= \overline{H^{\alpha}(\D_u)}H^{\beta}(u)\vert_{x = u = 0},
\end{equation*}
where the operator $\D_{u}$ on the right hand side of the expression denotes that each variable in $H^{\alpha}(u)$ is replaced by its 
corresponding partial derivative. 
The formal adjoint can then be obtained using integration by parts (this is the reason for choosing functions with compact support), 
as this allows us to shift the derivation in $x$. We thus have that $\inner{\D_u,\D_x}^{\ast} = -\pi_{u}[\inner{u,\D_x}]$, where the 
projection operator $\pi_{u}$ appears because one must stay in the space of harmonics in $u \in \mR^m$, 
$\pi_{u}[\inner{u,\D_x}]^{\ast} = -\inner{\D_u,\D_x}$ and finally also
\begin{equation*}
\bigg(\big(1 - S_v\Pi_{k+1} - C\Pi_{k-1}\big)\inner{v,\D_x}\bigg)^{\ast} = - \inner{\D_v,\D_x}.
\end{equation*}
In deriving this last result we made use of the fact that $S_v^{\ast} = S_u$ and $C^{\ast} = A$ (see \cite{DER2}) 
and both operators act trivially on the space of $\mcH_k$-valued fields. We can now put everything together, hereby making use of 
Branson's result (see Theorem \ref{theorem_Branson}) which tells us that there exist constants (determined in terms of Casimir eigenvalues) such that the second-order 
conformally invariant operator can be written as 
\begin{equation}\label{D_k_BR}
\begin{split}
\mcD_k & := c_{k,1}\inner{\D_v,\D_x}\big(1 - S_v\Pi_{k+1} - C\Pi_{k-1}\big)\inner{v,\D_x}  \\ 
& + c_{k+1}\inner{\D_u,\D_x} \pi_{u}[\inner{u,\D_x}] + c_{k-1}\pi_{u}[\inner{u,\D_x}] \inner{\D_u,\D_x}. 
\end{split}
\end{equation}
We will first recast this expression, still containing a few real unknown constants, and we will then fix these constants using the 
uniqueness of the second-order conformally invariant operator (see \cite{DER1}). First of all, we have the 
following (from direct calculations): 
\begin{equation*}
\big[\inner{\D_u,\D_x} , \pi_{u}[\inner{u,\D_x}]\big] = \Delta_x -\frac{2}{2\mE_x + m - 2}\pi_{u}[\inner{u,\D_x}]\inner{\D_u,\D_x}.
\end{equation*}
Note that the Euler operator in the denominator poses no threat here: it acts as a constant on our space of values. This already tells 
us that the last two terms in expression (\ref{D_k_BR}) for $\mcD_k$ can be written as 
\begin{equation*}
 p_k\Delta_x + q_k \pi_{u}[\inner{u,\D_x}] \inner{\D_u,\D_x},
 \end{equation*}
with $p_k$ and $q_k$ two real constants which will depend on $k$ due to the presence of the Euler operator in the lemma. As for 
the first term in $\mcD_k$, it suffices to note that the derivation in $v \in \mR^m$ can only act on the embedding factors $S_v$ 
and $C$, or on the $v$ appearing in the operator $\inner{v,\D_x}$ as the field $f(x,u)$ does not depend on this dummy variable. 
Taking into account that
\begin{align*}
[ \inner{\D_v,\D_x}, \inner{v,\D_x}] & =  \Delta_x\\
[\inner{\D_v,\D_x},S_v] & =  \inner{\D_u,\D_x}\\
 [\inner{\D_v,\D_x},C] & = \pi_{u}\inner{u,\D_x} ,
\end{align*}
there should thus exist an expression of the form 
\begin{equation*}
\mcD_k = \alpha_k\Delta_x + \beta_k\pi_{u}[\inner{u,\D_x}]\inner{\D_u,\D_x},
\end{equation*}
where the real constants $(\alpha_k,\beta_k)$ can now be fixed by requiring that the result is conformally invariant. 
This then fixes the operator $\mcD_k$ up to a multiplicative constant, as was done in \cite{DER1}.

\subsection{$k$ vector variables}

A similar argument can now be used to construct the generalisation of $\mcD_k$ to the case of $k$ vector variables. This leads to a 
second-order conformally invariant operator $\mcD_\lambda$ acting on functions taking values in the $\SO(m)$-module with 
highest weight $\lambda = (\lambda_1,\ldots,\lambda_k,0,\ldots,0)$ where $k < \lfloor \frac{m}{2} \rfloor$. A model for this representation is presented in the following definition:
\begin{definition}
A polynomial $P(u_1,\ldots,u_k)\in \mcP\brac{\mR^{km},\mC}$ in $k$ vector variables, is called simplicial harmonic 
if it satisfies the system
\begin{align*}
\inner{ \D_i, \D_j}P &= 0, \ \text{for all}\  i, j =1, \ldots, k \\
\inner{ u_i, \D_j} P &= 0, \ \text{for all}\ 1 \leq i <j \leq k.
\end{align*} 
\end{definition}
The vector space of simplicial polynomials homogeneous of degree $\lambda_i$ in the vector variable $u_i$ will be denoted by 
$\mcH_{\lambda}\brac{\mR^{km},\mC}$, where $\lambda = (\lambda_1, \ldots, \lambda_k)$. It is known (see e.g. \cite{GW, We}) that the algebra of 
$\SO(m)$-invariant operators acting on $\mcP\brac{\mR^{km},\mC}$ is the universal enveloping algebra of the symplectic Lie algebra 
\begin{align*}
\sym(2k) =& \Span\set{\inner{\D_i, \D_j}, \inner{u_i, u_j} , \inner{ u_i, \D_j} : 1 \leq i,j \leq k, i\neq j}\\
&\oplus\Span\set{\Delta_j,\norm{u_j}^2, \mE_j +\frac{m}{2}: 1\leq j \leq k}.
\end{align*}
We can then define the operator $\mcD_{\lambda}$ as follows:
\begin{definition}
The higher spin Laplace operator in $k$ vector variables is defined as the unique (up to a multiplicative constant) conformally 
invariant second-order differential operator 
$\mcD_{\lambda}:\mcC^{\infty}(\mR^m, \mcH_{\lambda}) \rightarrow \mcC^{\infty}(\mR^m, \mcH_{\lambda})$.
\end{definition}
To obtain an explicit expression for $\mcD_{\lambda}$, we will again justify a particular ansatz, i.e. a linear combination 
of transvector generators in terms of unknown constants, which will be determined by demanding conformal invariance. 
This time we start from the identification 
\begin{equation}\label{nabla_lambda}
\nabla f_\lambda \in \Gamma(\mcH_\lambda\otimes T^{\ast}M)\ \stackrel{1:1}{\longleftrightarrow}\ \sum_{i=1}^m\sum_{\alpha = 1}^{d_\lambda}\frac{\partial \varphi_\alpha(x)}{\partial x_i}v_iH_\lambda^\alpha(u) \in \mcC^\infty(\mR^m,\mcH_\lambda \otimes \mcH_1)
\end{equation}
where $\alpha \in \{1,\ldots,d_\lambda\}$ is an index labeling basis vectors for the $d_\lambda$-dimensional space 
$\mcH_\lambda(\mR^{k m},\mC)$. In the present setting one can use the transvector algebra
\begin{equation*}
Z(\sym(2k + 2),\sym(2k) \oplus \spl(2)),
\end{equation*} 
where the symplectic algebras are generated by either all the invariant operators or the ones expressed in terms 
of the dummy variables in the definition for $\sym(2k)$, and where $\spl(2)$ is generated by $\Delta_v$ and $|v|^2$ 
(in a new dummy variable $v \in \mR^m$). 
The generators for this transvector algebra are projections of all the operators in the bigger symplectic Lie algebra 
$\sym(2k + 2)$ which are not present in the sum $\sym(2k) \oplus \spl(2)$; these are precisely all operators which 
contain both $v$ and a dummy variable $u_j$ with $1 \leq j \leq k$ (either as a variable, or as a differential operator). 
Using the same notations as before, we label these generators as
\begin{equation*}
\big\{S_u^{(j)},S_v^{(j)},A_j,C_j : 1 \leq j \leq k\big\},
\end{equation*}
where, for instance, one has that $C_j = \pi_{\sym(2k)}\pi_{\spl(2)}[\inner{ u_j,v}]$ with $\pi_\g$ the extremal projection operator for the 
Lie algebra $\g$. For $\g = \spl(2)$ and $\g = \sym(2k)$, these operators are commuting, so the order of the projection operators is 
irrelevant. The generalised gradient method then again reduces to investigating expressions of the following form 
(where we use $U = (u_1,\ldots,u_k)$ as a short-hand notation): 
\begin{equation}\label{vj_lambda}
v_iH^\alpha_\lambda(U) = H_{\lambda,1}(U;v) + \sum_{j = 1}^k \brac{S^{(j)}_v H_{\lambda + \varepsilon_j}(U) + C_j 
H_{\lambda-\varepsilon_j}(U)},
\end{equation}
where $H_{\lambda\pm \varepsilon_j}(U)\in \mcH_{\lambda\pm\varepsilon_j}$ and where 
$H_{\lambda,1}(U;v)\in \mcH_{\lambda,1}$. Here $\lambda \pm \varepsilon_j$ stands for the highest weight obtained by adding 
or removing one from the entry $\lambda_j$. Note that if this weight is not dominant, it is not taken into account. Also, the 
right-hand side still depends on $(\alpha,i)$ but we again decided to suppress these notations as one needs to sum over these 
indices after all, which then gives
\begin{equation*}
\nabla f_\lambda(x,U) = f_{\lambda,1}(x,U,v) + \sum_{j = 1}^k \brac{S^{(j)}_v f_{\lambda + \varepsilon_j}(x,U) + 
C_j f_{\lambda-\varepsilon_j}(x,U)}.
\end{equation*}
One can now try to obtain the generalised gradients by working out all the projections. 
First of all, consider the action of a generator 
\begin{equation*}
A_j := \pi_{\sym(2k)}\pi_{\spl(2)}[\inner{\D_j,\D_v}] \in Z(\sym(2k+ 2),\sym(2k) \oplus \spl(2)),
\end{equation*}
where $1 \leq j \leq k$. Acting with this operator on the left-hand side of (\ref{vj_lambda}) gives a polynomial in the dummy variables $U$. Indeed, as the operator $\pi_{\spl(2)}$ reduces to the identity operator here, we get
\begin{equation*}
\pi_{\sym(2k)}\big[\inner{ e_i,\D_{j}} H^\alpha_\lambda(u_1,\ldots,u_k)\big] \in \mcH_{\lambda-\varepsilon_j}(\mR^{km},\mC).
\end{equation*}
The fact that the result belongs to $\mcH_{\lambda-\varepsilon_j}(\mR^{km},\mC)$ is crucial here, and 
follows from the observation that it has the desired degree of homogeneity in the dummy variables, and that 
the action of $\pi_{\sym(2k)}$ projects it on the kernel of the differential operators defining this vector space 
of simplicial harmonics. But this tells us that the action of the operator $A_j = \pi_{\sym(2k)}\pi_{\spl(2)}[\inner{\D_{j},\D_v}]$ 
on the right-hand side of (\ref{vj_lambda}) belongs to this space too. Indeed, there is a unique summand transforming 
as an element of the representation with highest weight $(\lambda - \varepsilon_j)$, so they have to be equal. 
We now claim that there exists a constant $c_j$ such that 
\begin{equation*}
A_jC_j f_{\lambda-\varepsilon_j}(x,U) = c_j f_{\lambda - \varepsilon_j}(x,U).
\end{equation*}
To see this, one can either invoke the quadratic relations between the generators defining the algebra 
$Z(\sym(2k + 2),\sym(2k) \oplus \spl(2))$, which we will not do here (although it would gives us the explicit value of the constant), 
or one can exploit the fact that 
\begin{equation*}
A_jC_j f_{\lambda-\varepsilon_j}(x,U) = \pi_{\sym(2k)}[\inner{\D_{j},\D_x}] f_\lambda(x,U),
\end{equation*}
which tells us that the operator $A_jC_j$ is a rotationally invariant endomorphism on $\mcH_{\lambda - \varepsilon_j}(\mR^{\ell m},\mC)$. Therefore, it has to be a multiple of the identity operator in view of Schur's lemma. This observation also allows us to conclude that 
\begin{equation*}
[A_j,C_a]f_{\lambda+\varepsilon_a}(x,U) = 0 = [A_k,C_b]f_{\lambda-\varepsilon_b}(x,U) 
\end{equation*}
for all $1 \leq a \leq k$ and $1 \leq b \neq j \leq k$, without knowing the explicit commutation relations. 
Reintroducing the summations over $\alpha$ and $i$ appearing in (\ref{nabla_lambda}), this then dictates the existence 
of a constant $\gamma_j$ for which
\begin{equation*}
\Pi_{\lambda - \varepsilon_j}\big(\nabla f_\lambda(x,U)\big) = \gamma_j\pi_{\sym(2k)}[\inner{\D_{j}, \D_x}]f_\lambda(x,U).
\end{equation*}
Note that the constant $\gamma_j$ can be zero, which merely means that the projection on this component in the decomposition 
happens to be trivial. A similar calculation can then be done for the irreducible summands corresponding to highest weights 
of the form $(\lambda + \varepsilon_j)$ in the decomposition (\ref{vj_lambda}), 
for which we consider the action of the generator 
\begin{equation*}
 S_u^{(j)} := \pi_{\sym(2k)}\pi_{\spl(2)}[\inner{ u_j,\D_{v}}] \in Z(\sym(2k + 2),\sym(2k) \oplus \spl(2))
\end{equation*}
with $1 \leq j \leq k$. The action on the left-hand side then gives
\begin{equation*}
\pi_{\sym(2k)}\big[\inner{ e_i,u_j}H^\alpha_\lambda(u_1,\ldots,u_k)\big] \in \mcH_{\lambda+\varepsilon_j}(\mR^{km},\mC).
\end{equation*}
As the summand $\mcH_{\lambda+\varepsilon_j}(\mR^{km},\mC)$ is unique in the decomposition (\ref{vj_lambda}), 
there must again in view of Schur's lemma exist a real constant $\beta_j$ such that the following holds: 
\begin{equation*}
\Pi_{\lambda + \varepsilon_j}\big(\nabla f_\lambda(x,U)\big) = \beta_j\pi_{\sym(2k)}[\inner{ u_j,\D_x }]f_\lambda(x,U). 
\end{equation*}
The final operator is then defined through subtraction: 
\begin{equation*}
\Pi_{\lambda,1}\left(\nabla f_\lambda(x,U)\right)
= \bigg(1 - \sum_{j=1}^k S^{(j)}_v\Pi_{\lambda+\varepsilon_j} - \sum_{j=1}^k C_j\Pi_{\lambda-\varepsilon_j}\bigg)
\inner{ v,\D_x}f_\lambda(x,U).
\end{equation*}
Now that we have an expression for the Stein-Weiss generalised gradients, we determine their formal adjoints on the space 
$\mcC_c(\mR^m,\mcH_\lambda)$. This is again done in terms of the following inner product (e.g. on basis elements)
\begin{equation*}
\inner{\varphi_\alpha(x) \otimes H^\alpha(u),\psi_\beta(x) \otimes H^\beta(u)} := \int_{\mR^m}\varphi_{\alpha(x)}
\overline{\psi_{\beta(x)}}\mathrm{d}x \otimes [H^\alpha(u),H^\beta(u)]_F.
\end{equation*}
An easy argument involving integration by parts shows that 
\begin{equation*}
\big(\pi_{\sym(2k)}\inner{\D_{j},\D_x}\big)^{\ast} =-\pi_{\sym(2k)}\big[\inner{ u_j,\D_x}] 
\end{equation*}
and vice versa (because $u_j$ and $\D_{j}$ are Fischer duals), and that 
\begin{equation*}
\brac{\Big(1 - \sum_{j=1}^k S^{(j)}_v\Pi_{\lambda+\varepsilon_j} - \sum_{j=1}^k C_j\Pi_{\lambda-\varepsilon_j}\Big)
\inner{ v,\D_x}}^{\ast} = -\inner{ \D_v,\D_x}.
\end{equation*}
This means that the invariant operator $\mcD_\lambda$ can (up to a multiplicative constant) be written as an operator of the form
\begin{equation*}
\begin{split}
\mcD_\lambda &= c_{\lambda,1}\inner{ \D_v,\D_x}\big(1 - \sum_{j=1}^k S^{(j)}_v\Pi_{\lambda+\varepsilon_j} - \sum_{j=1}^k C_j\Pi_{\lambda-\varepsilon_j}\big)\inner{ v,\D_x} \\
&+ \sum_{j = 1}^k c_{\lambda + \varepsilon_j}\pi_{\sym(2k)}\inner{ \D_{j},\D_x}\pi_{\sym(2k)}\big[\inner{u_j,\D_x}]\\
&+ \sum_{j = 1}^k c_{\lambda - \varepsilon_j}\pi_{\sym(2k)}\big[\inner{ u_j,\D_x}]\pi_{\sym(2k)}\inner{\D_{j},\D_x}.
\end{split}
\end{equation*}
In case of a non-dominant weight, the corresponding operators are to be omitted from the summation. In the next section, we will 
first simplify this expression to obtain a form similar as the one we found for $\mcD_k$ after which we will compute the unknown 
constants by explicitly proving conformal invariance.


\section{Construction of $\mcD_{\lambda}$}

In order to find an explicit expression for the higher spin Laplace operator $\mcD_\lambda$, we need to introduce the extremal 
projection operator $\pi_{\sym(2k)}$ for the simple Lie algebra $\sym(2k)$. When we choose the positive root vectors of 
this algebra to be $\inner{\D_i, \D_j}$ with $1 \leq i, j \leq k$ and $\inner{ u_a, \D_b}$ with $1 \leq a < b \leq k$, 
the extremal projection operator for the Lie algebra $\sym(2k)$ will project on the intersection of the kernels of all these positive root vectors, and this is then precisely the space of simplicial harmonics. 
The general construction can be found in e.g. \cite{Molev}. An explicit construction in the setting of higher spin Dirac operators 
has been made in \cite{ER1}, but can easily be adapted to this setting. The construction itself can also be found in \cite{Tim}, 
whence we will just mention the operator here. 
\begin{definition}
An ordering of a set of positive roots for a Lie algebra is called normal if any composite root lies between its components.
\end{definition}

Moreover, we have the following Lemma (see e.g. \cite{Molev}).
\begin{lemma}
For each Lie algebra, there exists a normal ordering on the set of positive roots which is not necessarily unique.
\end{lemma}

In case of the symplectic Lie algebra $\sym(2k)$, we for example have the following normal ordering:
\begin{align*}
\varepsilon_1-\varepsilon_2,\varepsilon_1-&\varepsilon_3,\ldots,\varepsilon_1-\varepsilon_k, \varepsilon_2-\varepsilon_3,\varepsilon_2-\varepsilon_4,\ldots,\varepsilon_{k-1}-\varepsilon_k, \\
-2&\varepsilon_k,\ldots,-2\varepsilon_3,-\varepsilon_2-\varepsilon_3,-\varepsilon_1-\varepsilon_3,-2\varepsilon_2,
-\varepsilon_1-\varepsilon_2,-2\varepsilon_1.
\end{align*}
Next, we define the operators 
\begin{align*}
\pi_{-2\varepsilon_a} &=\sum_{s=0}^{+\infty} \frac{1}{4^s s!} 
\frac{\Gamma(-\mE_a-\frac{m}{2}+a+1)}{\Gamma(-\mE_a-\frac{m}{2}+a+1+s)}|u_a|^{2s} \Delta_a^s \qquad 1 \leq a \leq k \\
\pi_{-\varepsilon_a-\varepsilon_b} &= \sum_{s=0}^{+\infty} \frac{1}{s!}
\frac{\Gamma(-\mE_a-\mE_b-m+a+b+1)}{\Gamma(-\mE_a-\mE_b-m+a+b+s+1)}\inner{u_a, u_b}^s \inner{\D_a, \D_b}^s 
\qquad 1 \leq a,b \leq k \\
\pi_{\varepsilon_i-\varepsilon_j}&=\sum_{s=0}^{+\infty} \frac{(-1)^s}{s!}\frac{\Gamma(\mE_i-\mE_j+j-i+1)}{\Gamma(\mE_i-\mE_j+j-i+1+s)}\inner{u_j, \D_i}^s \inner{u_i, \D_j}^s \qquad 1 \leq i< j \leq k.
\end{align*}
Taking the product of all these operators in any normal ordering gives us the extremal projector $\pi_{\sym(2k)}$. 
\begin{proposition}\label{prop_extremal_sym(2k)}
The extremal projection operator $\pi_{\sym(2k)}$ for $\sym(2k)$ is given by the following product of operators: 
\begin{align*}
\pi_{\sym(2k)}=&\pi_{\varepsilon_1-\varepsilon_2}\ldots\pi_{\varepsilon_1-\varepsilon_k}\pi_{\varepsilon_2-\varepsilon_3}\ldots
\pi_{\varepsilon_{k-1}-\varepsilon_k}\pi_{-2\varepsilon_k}\ldots\pi_{-2\varepsilon_3}\\
&\times\pi_{-\varepsilon_2-\varepsilon_3}\pi_{-\varepsilon_1-\varepsilon_3}\pi_{-2\varepsilon_2}\pi_{-\varepsilon_1-\varepsilon_2}
\pi_{-2\varepsilon_1}.
\end{align*}
\end{proposition}

This operator has the properties that
\begin{align}
\inner{ u_i, \D_j} \pi_{\sym(2k)} &= 0 \qquad \forall 1 \leq i < j \leq k \label{rel1}\\
\inner{ \D_i, \D_j}\pi_{\sym(2k)} &= 0 \qquad \forall 1 \leq i, j \leq k  \label{rel2}\\
\pi_{\sym(2k)} \inner{u_j, \D_i}&= 0 \qquad \forall 1 \leq i < j \leq k \label{rel3} \\
\pi_{\sym(2k)} \inner{u_i, u_j} &= 0 \qquad \forall 1 \leq i, j \leq k. \label{rel4}
\end{align}
\begin{remark}
When acting on polynomials, the extremal projector simplifies a lot, since the finiteness of the degrees in each vector variable causes the terms of each infinite series sum to act trivially above a certain index. 
\end{remark}

Similar to the construction of twistor operators in \cite{ER1}, it can be shown that the generalised gradients 
\begin{align*}
G_{\lambda,\lambda+\varepsilon_j}:\mcC^\infty(\mR^m, \mcH_\lambda)&\longrightarrow \mcC^\infty(\mR^m, \mcH_{\lambda+\epsilon_j})\\
G_{\lambda,\lambda+\varepsilon_j}^{\ast}:\mcC^\infty(\mR^m, \mcH_{\lambda+\epsilon_j})&\longrightarrow \mcC^\infty(\mR^m, \mcH_{\lambda}).
\end{align*}
are equal to 
\begin{equation*}
G_{\lambda,\lambda+\varepsilon_j}:=\pi_{\sym(2k)}\inner{u_j,\D_x}\quad \text{and}\quad 
G_{\lambda,\lambda+\varepsilon_j}^{\ast}:=\pi_{\sym(2k)}\inner{\D_j,\D_x}.
\end{equation*}
Note that the expressions for these operators are independent of $\lambda$ as the precise form for a particular choice of $\lambda$ is absorbed into the action of $\pi_{\sym(2k)}$, which contains Euler operators generating the appropriate constants. This means that we may use the shorthand notations
$G_{\varepsilon_j}:=G_{\lambda,\lambda+\varepsilon_j}$ and similarly for its formal adjoint. The operator $G_{\varepsilon_j}^{\ast}$
can be simplified to
\begin{equation*}
G_{\varepsilon_j}^{\ast}= \prod_{p=j+1}^k \left( 1-\frac{\inner{ u_p, \D_j} \inner{ u_j, \D_p}}{\mE_j-\mE_p+p-j+1} \right) \inner{\D_j,\D_x},
\end{equation*}
or, when using $\sym(2k+2)$ commutation relations, the PBW theorem and the fact that this operator acts on 
$\mcH_\lambda$-valued functions,
\begin{equation*}
G_{\varepsilon_j}^{\ast} = \inner{\D_j,\D_x} + \sum_{j < i_1 < \cdots < i_s \leq k} \frac{\inner{ u_{i_1}, \D_j} 
\inner{ u_{i_2}, \D_{i_1}}\cdots \inner{u_{i_s}, \D_{i_{s-1}}}\inner{ \D_{i_s}, \D_x}}{(\mE_j-\mE_{i_1}+i_1-j) \cdots 
(\mE_j-\mE_{i_s}+i_s-j)}.
\end{equation*} 
\begin{proposition}
The composition of the twistor operator composed with its formal adjoint can be simplified to
\begin{equation*}
G_{\varepsilon_a}^{\pha}G_{\varepsilon_a}^{\ast}=  \pi_{\sym(2k)}\brac{\sum_{j=a}^k a_j \inner{ u_j, \D_x}\inner{ \D_j, \D_x}},
\end{equation*}
for some coefficients $a_j\in \mcU'(\h)$. 
\end{proposition}
Although the right hand-side of this expression is less canonical, it is easier to use in practical computations. 
\begin{proof}
Follows from $[\inner{ u_j, \D_x}, \inner{u_i, \D_j}] = -\inner{ u_i, \D_x }$, and equations (\ref{rel3}) and (\ref{rel4}). 
\end{proof}

A similar argument can be used to obtain the following result:
\begin{proposition}
The following equality holds for the opposite order:
\begin{equation*}
G_{\varepsilon_a}^{\ast}G_{\varepsilon_a}^{\pha}=\pi_{\sym(2k)}\brac{ \sum_{j=1}^k \brac{b_j  \inner{ \D_j, \D_x}\inner{ u_j, \D_x}+d_j \inner{ u_j, \D_x}\inner{ \D_j, \D_x}}},
\end{equation*}
for some coefficients $b_j, d_j\in \mcU'(\h)$.
\end{proposition}

Using these propositions, we expect the higher spin Laplace operator to be of the form:
\begin{equation}
\label{Dlambdaprop}
\mcD_\lambda:= \Delta_x +  \pi_{\sym(2k)}\brac{\sum_{p=1}^k c_p\inner{u_p, \D_x}\inner{\D_p, \D_x}},
\end{equation}
where the constants $c_p$ are coefficients in terms of Euler operators which have to be determined. 
These constants have to be chosen such that $\mcD_\lambda$ itself is conformally invariant. 
This means that infinitesimal rotations, translations, dilations and the special conformal transformations should be generalised 
symmetries for this operator, see also \cite{DER1} for more details. Independent of the choice of the $c_p$, we already have the following properties: this operator is clearly invariant under the infinitesimal rotations given by
\begin{equation*}
dL(e_{ij}):=L_{ij}^x + \sum_{p=1}^k L_{ij}^{u_p}.
\end{equation*}
Indeed, $[\inner{\D_{p}, \D_x}, dL(e_{ij})] = 0$, $[\inner{u_p, \D_x}, dL(e_{ij})] = 0$, and $dL(e_{ij})$ commutes with all elements 
of the operator algebra $\sym(2k)$. All terms of $\mcD_\lambda$ obviously commute with $\D_{x_j}$, the infinitesimal translations, 
as only the dummy variables and the operators $\D_{x_i}$ appear. Since $\mcD_\lambda$ is a homogeneous second-order 
differential operator, we have 
\begin{equation*}
\mcD_{\lambda} \left( \mE_x + \frac{m-2}{2} \right) = \left( \mE_x + \frac{m+2}{2} \right)\mcD_{\lambda},
\end{equation*} 
so that also infinitesimal dilations are generalised symmetries. The only thing that remains to be proved is that the special conformal 
transformations are generalised symmetries. The inversion $\mcJ_{\lambda}$ for any function $f$ is defined as follows 
(see appendix \ref{Appendix_inversion}):
\begin{equation}\label{Harmonic_inversion}
\mcJ_{\lambda} f(x, u_1, \ldots, u_k) = |x|^{2w} f\left( \frac{x}{|x|^2}, \frac{x u_1 x}{|x|^2}, \ldots, \frac{x u_k x}{|x|^2} \right),
\end{equation}
where $w$ is the conformal weight for the higher spin Laplace operator, namely $w=1-\frac{m}{2}$.
Special conformal transformations are given by the consecutive action of this inversion, an infinitesimal translation 
and another inversion, so we want to prove that there exists an operator $d$ such that 
$\mcD_\lambda \mcJ_\lambda \partial_{x_j} \mcJ_\lambda = d \mcD_\lambda$. 
This comes down to determining the constants $c_p$ such that this operator $d$ exists. First of all, we have that 
\begin{equation*}
\mcJ_\lambda \D_{x_j} \mcJ_\lambda = |x|^2 \partial_{x_j}-x_j (2 \mE_x + m-2)+2 \sum_{p=1}^k \inner{ u_p, x} \D_{u_{pj}} - 
2 \sum_{p=1}^k u_{pj} \inner{ x, \D_p}.
\end{equation*}
We will calculate the action of each of the terms appearing in \eqref{Dlambdaprop} on these special conformal transformations which requires some technical lemmas throughout the pages below. First of all, note that the following relations hold:
\begin{align*}
\Delta_x |x|^2 &= 2m + 4 \mE_x + |x|^2 \Delta_x \\
\Delta_x x_j &= 2 \D_{x_j} + x_j \Delta_x \\
\Delta_x \inner{ u_p, x } &= 2 \inner{ u_p, \D_x } + \inner{ u_p, x } \Delta_x \\
\Delta_x \inner{ x, \D_p } &= 2 \inner{ \D_p, \D_x } + \inner{ x, \D_p } \Delta_x.
\end{align*}
Using these relations, we obtain our first technical lemma. The proof can be obtained through tedious but straightforward computations.
\begin{lemma}
The following relation holds: 
\begin{equation*}
 \Delta_x \mcJ_\lambda \D_{x_j} \mcJ_\lambda = \left( \mcJ_\lambda \D_{x_j} \mcJ_\lambda - 4x_j \right) \Delta_x + 4 \sum_{p=1}^k \inner{ u_p, \D_x} \D_{u_{pj}} - 4 \sum_{p=1}^k u_{pj} \inner{ \D_p, \D_x }.
\end{equation*}
\end{lemma}
From this result, we see that $\mcJ_\lambda \D_{x_j} \mcJ_\lambda - 4x_j$ will be the candidate for the operator $d$, obviously, if the two  remaining summations can be canceled out by the action of the other terms of \eqref{Dlambdaprop} on the special conformal transformations. 
\begin{lemma}
The following relation holds: 
\begin{align*}
\pi_{\sym(2k)}\inner{ u_a, \D_x } \inner{ \D_a, \D_x } \mcJ_\lambda \D_{x_j} \mcJ_\lambda =& 
\pi_{\sym(2k)}\inner{ u_a, \D_x }\Big( (\mcJ_\lambda \D_{x_j} \mcJ_\lambda -2x_j)\inner{ \D_a, \D_x } \\
 &+ (m+2\mE_a-2a+2) \D_{u_{aj}} +2 \sum_{p=a+1}^k \inner{ u_p, \D_a } \D_{u_{pj}}\Big).
\end{align*}
\end{lemma}
\begin{proof}
To compute this expression, we use the following relations:
\begin{align*}
\inner{ \D_a, \D_x } |x|^2 \D_{x_j} &= 2\inner{ x, \D_a} \D_{x_j} + |x|^2 \D_{x_j} \inner{ \D_a, \D_x } \\
\inner{ \D_a, \D_x } x_j &= \D_{u_{aj}} + x_j \inner{ \D_a, \D_x } \\
\inner{ \D_a, \D_x } \inner{ u_a, x } \D_{u_{aj}} &= (m+\mE_a + \mE_x) \D_{u_{aj}} + \inner{ u_a, x } \D_{u_{aj}} \inner{ \D_a, \D_x }\\
\inner{ \D_a, \D_x } u_{aj} \inner{ x, \D_a } &= \D_{u_{aj}} + \inner{ x, \D_a } \D_{x_j} + u_{aj} \Delta_a + u_{aj} \inner{ x, \D_a } \inner{ \D_a, \D_x }.
\end{align*}
Also, if $a \neq p$, we have the following two identities:
\begin{align*}
\inner{ \D_a, \D_x } \inner{ u_p, x } \D_{u_{pj}} &= \inner{ u_p, \D_a } \D_{u_{pj}} + \inner{ u_p, x } \D_{u_{pj}} \inner{ \D_a, \D_x } \\
\inner{ \D_a, \D_x } u_{pj} \inner{ x, \D_p } &= u_{pj} \inner{ \D_a, \D_p } + u_{pj} \inner{ x, \D_p } \inner{ \D_a, \D_x }.
\end{align*}
Using these relations, we can rewrite $\pi_{\sym(2k)}\inner{ u_a, \D_x } \inner{ \D_a, \D_x } \mcJ_\lambda \D_{x_j} \mcJ_\lambda$ as follows:
\begin{align*}
&\pi_{\sym(2k)}\inner{ u_a, \D_x } \inner{ \D_a, \D_x } \mcJ_\lambda \D_{x_j} \mcJ_\lambda \\
&=\pi_{\sym(2k)}\inner{ u_a, \D_x } \bigg( 2\inner{ x, \D_a} \D_{x_j} + |x|^2 \D_{x_j} \inner{ \D_a, \D_x } - (2 \mE_x+m-2)\D_{u_{aj}} \\ 
&- x_j(2\mE_x+m-2+2)\inner{ \D_a, \D_x } +2(m+\mE_a + \mE_x) \D_{u_{aj}} + 2 \sum_{p\neq a,p=1}^k 
\inner{ u_p, \D_a } \D_{u_{pj}} - 2\D_{u_{aj}} \\ 
&+ 2 \sum_{p=1}^k \inner{ u_p, x } \partial_{u_{pj}} \inner{ \D_a, \D_x } 
 - 2\inner{ x, \D_a } \D_{x_j} - 2 \sum_{p=1}^k u_{pj} \inner{ \D_a, \D_p } - 2 \sum_{p=1}^k u_{pj} \inner{ x, \D_p } \inner{ \D_a, \D_x } \bigg).
\end{align*}
Remember that from the definition of $\mcD_\lambda$, the functions on which all of these operators act are elements of 
$\mcC^\infty(\mR^m, \mcH_\lambda)$. This means that any term ending with a defining operator of $\mcH_\lambda$ 
acts trivially and can be omitted. This leads to the desired result.
\end{proof}
\begin{lemma}
The following relation holds:
\begin{align*}
\pi_{\sym(2k)}&\inner{ u_a, \D_x }\left[(\mcJ_\lambda \D_{x_j} \mcJ_\lambda -2x_j)\inner{ \D_a, \D_x }\right]\\
=&\ \pi_{\sym(2k)}(\mcJ_\lambda \D_{x_j} \mcJ_\lambda -4x_j)\inner{ u_a, \D_x } \\
&+ \pi_{\sym(2k)}\bigg(u_{aj}(-m-2\mE_a+2a-2) -2 \sum_{p=a+1}^k u_{pj} \inner{ u_a, \D_p } \D_{u_{pj}}\bigg)\inner{ \D_a, \D_x }.
\end{align*}
\end{lemma}
\begin{proof}
To compute this expression, we can use the relations
\begin{align*}
\inner{ u_a, \D_x } |x|^2 \D_{x_j} &= 2\inner{ x, u_a} \D_{x_j} + |x|^2 \D_{x_j} \inner{ u_a, \D_x } \\
\inner{ u_a, \D_x } x_j &= u_{aj} + x_j \inner{ u_a, \D_x } \\
\inner{ u_a, \D_x } \inner{ u_a, x } \D_{u_{aj}} &= |u_a|^2 \D_{u_{aj}} -\inner{ u_a, x } \D_{x_j} + \inner{ u_a, x } \D_{u_{aj}} \inner{ u_a, \D_x }\\
\inner{ u_a, \D_x } u_{aj} \inner{ x, \D_a } &= u_{aj}(\mE_a-\mE_x) + u_{aj} \inner{ x, \D_a } \inner{ u_a, \D_x }
\end{align*}
and also, if $a \neq p$, the following identities hold:
\begin{align*}
\inner{ u_a, \D_x } \inner{ u_p, x } \D_{u_{pj}} &= \inner{ u_a, u_p } \D_{u_{pj}} + \inner{ u_p, x } \D_{u_{pj}} \inner{ u_a, \D_x } \\
\inner{ u_a, \D_x } u_{pj} \inner{ x, \D_p } &= u_{pj} \inner{ u_a, \D_p } + u_{pj} \inner{ x, \D_p } \inner{ u_a, \D_x }.
\end{align*}
Using all of these, we obtain
\begin{align*}
\pi_{\sym(2k)}&\inner{ u_a, \D_x }\left[(\mcJ_\lambda \D_{x_j} \mcJ_\lambda -2x_j)\inner{ \D_a, \D_x }\right] \\
=& \pi_{\sym(2k)} \bigg[ 2\inner{ x, u_a} \D_{x_j} + |x|^2 \D_{x_j} \inner{ u_a, \D_x } - u_{aj}(2\mE_x+m) - x_j(2\mE_x+m+2) \inner{ u_a, \D_x }\\
& +2\sum_{p=1}^k \inner{ u_p, x }\D_{u_{pj}} \inner{ u_a, \D_x } +2 \sum_{p=1}^k \inner{ u_a, u_p } \D_{u_{pj}} -2 \inner{ u_a, x } \D_{x_j}\\
&-2 \sum_{p=1}^k u_{pj} \inner{ x, \D_p } \inner{ u_a, \D_x } - 2\sum_{p\neq a, p=1}^k u_{pj} \inner{ u_a, \D_p } -2 u_{aj}(\mE_a-\mE_x) \bigg] \inner{ \D_a, \D_x }.
\end{align*}
Using the fact that $\pi_{\sym(2k)}$ annihilates any negative root vector the expression above can be simplified to complete the proof.
\end{proof}
Altogether, we find
\begin{align*}
\mcD_\lambda\mcJ_\lambda \D_{x_j} \mcJ_\lambda=&(\mcJ_\lambda \D_{x_j} \mcJ_\lambda -4x_j) \mcD_\lambda \\
&+ 4 \sum_{p=1}^k \inner{ u_p, \D_{x}} \D_{u_{pj}} - 4 \sum_{p=1}^k u_{pj} \inner{ \D_p, \D_x } \\
&+ \pi_{\sym(2k)} \sum_{a=1}^k c_a\left( \inner{ u_a, \D_x } (m+2\mE_a-2a+2) \D_{u_{aj}} - 2 \sum_{p=a+1}^k \inner{ u_p, \D_x } \D_{u_{pj}} \right. \\
&\left. + u_{aj}(-m-2\mE_a+2a-2)\inner{ \D_a, \D_x } +2 \sum_{p=a+1}^k u_{pj} \inner{ \D_p, \D_x } \right).
\end{align*}
Note that $\Delta_x = \pi_{\sym(2k)} \Delta_x$. Also, the coefficients of $\inner{ u_b, \D_x } \D_{u_{bj}}$ and 
$u_{bj} \inner{ \D_b, \D_x }$ are always the opposite of each other. We want all coefficients of these terms to be zero. This means that for the coefficient of $\inner{ u_1, \D_x } \D_{u_{1j}}$, we obtain
\begin{equation*}
4+c_1(2\mE_1+m-2) = 0,
\end{equation*}
Hence, 
\begin{equation*}
c_1 = \frac{-4}{2\mE_1+m-2}.
\end{equation*}
For the coefficient of $\inner{ u_2, \D_x } \D_{u_{2j}}$, we find
\begin{equation*}
4-2c_1+c_2(m+4\mE_2-4) = 0,
\end{equation*}
or, after substituting the expression for $c_1$,
\begin{equation*}
c_2 = \frac{-4(2\mE_1+m)}{(2 \mE_2+m-4)(2\mE_1+m-2)}.
\end{equation*}
Continuing inductively, this leads to 
\begin{equation*}
4-2c_1-2c_2- \cdots - 2c_{p-1}+c_p(2\mE_p+m-2p) = 0,
\end{equation*}
whence $c_p$ is given by
\begin{equation*}
c_p = \frac{-4}{2\mE_p+m-2p}\prod_{j=1}^{p-1}\frac{2\mE_j+m-2j+2}{2\mE_j+m-2j}.
\end{equation*}
We finally arrive at the following theorem.
\begin{theorem}
The higher spin Laplace operator in $k$ vector variables is explicitly given by the following formula:
\begin{equation*}
\mcD_\lambda:= \pi_{\sym(2k)}\left(\Delta_x + \sum_{p=1}^k c_p \inner{ u_p, \D_x } \inner{ \D_p, \D_x } \right),
\end{equation*}
with $\pi_{\sym(2k)}$ given in proposition \ref{prop_extremal_sym(2k)} and where the constants $c_p$ are given by
\begin{equation*}
c_p = \frac{-4}{2\mE_p+m-2p}\prod_{j=1}^{p-1}\frac{2\mE_j+m-2j+2}{2\mE_j+m-2j}.
\end{equation*}
\end{theorem}


\section{Type A solutions of higher spin Laplace operators} 

As in any function theory, the study of polynomial solutions for the differential operator under consideration plays a crucial role. This is due to the fact that these 
are often used to decompose arbitrary solutions. We will therefore take a closer look at this problem in the case of the higher spin Laplace operator. In contrast to 
the case of the Laplace operator, for which homogeneous polynomial solutions of a fixed degree form an irreducible module for the orthogonal Lie algebra, the 
space of polynomial solutions for more complicated second order operators decomposes into a direct sum of many components. These can typically be divided in 2 
categories: solutions of type A (which can be seen as solutions satisfying extra gauge conditions) and solutions of type B (which are `induced' from type A 
solutions for a related operator by the action of a dual twistor operator). In this section we will therefore study these basic building blocks for general solutions, the 
{\em type A solutions}, hereby following the ideas used for the higher spin Dirac operators in \cite{DSER}. 
\begin{definition}\label{TypeA}
For a fixed integer highest weight $\lambda = (\lambda_1,\ldots,\lambda_k,0,\ldots,0)$ and an integer $\lambda_0 \geq \lambda_1$, we define the space 
\begin{equation*}
\ker_{\lambda_0}^A\mcD_\lambda := \mcP_{\lambda_0}(\mR^m,\mcH_\lambda) \cap \ker \brac{\Delta_x, \langle \partial_i, \partial_x \rangle: 1 \leq i \leq k}\ .
\end{equation*}
\end{definition}
It immediately follows from the definition of $\mcD_\lambda$ that $\ker_{\lambda_0}^A\mcD_\lambda$ is indeed a subspace of $\ker \mcD_\lambda$, containing solutions which are homogeneous in $x \in \mR^m$ of degree $\lambda_0$. Each type A solution is a polynomial in $(k+1)$ vector variables, the unknown $x$ and $k$ dummy variables $u_j$ fixing the values. From now on, we will often identify $x \equiv u_0$. This notation allows us to treat all variables (both $x$ and the dummies) on the same footing. We can then for instance consider the symplectic Lie algebra $\sym(2k+2)$ in the polynomial model from the previous section, with indices running from $0$ to $k$. As such, it is easily seen that the type A solutions are a subspace of the following space: 
\begin{definition}
The space of Howe harmonics in $(k+1)$ vector variables is defined as the subspace of all polynomials in $\mcP(\mR^{(k+1)m},\mC)$ which are annihilated by each of the constant coefficient differential operators in $\sym(2k+2)$. In other words: 
\begin{equation*}
\mcH^h\brac{\mR^{(k+1)m},\mC}:=\mcP\brac{\mR^{(k+1)m},\mC}\cap \ker\brac{\Delta_x,\inner{\D_i,\D_x},\inner{\D_i, \D_j} : 1 \leq i, j \leq k}\ .
\end{equation*}
\end{definition}
Note that these operators correspond to the positive roots $-\varepsilon_a -\varepsilon_b$, with $0 \leq a, b \leq k$. Also note that if we want to consider homogeneous subspaces, we use the notation $\mcH^h_{\lambda_1,\ldots,\lambda_k}\brac{\mR^{(k+1)m},\mC}$ for the space consisting of polynomials of 
degree $\lambda_i$ in $u_i$. \\ 
The space $\mcH^h\brac{\mR^{(k+1)m},\mC}$ is obviously not irreducible, but decomposes into irreducible modules under the action of $\so(m)$. Since we are 
working with a polynomial model in $(k+1)$ variables, it is reasonable to expect a decomposition into simplicial harmonics of the form 
$\mcH_{\mu}(\mR^{(k+1)m},\mC)$, where we use the notation $\mu=(\mu_0,\mu_1,\ldots,\mu_k)$ for a dominant weight with $(k+1)$ non-trivial integer 
entries. As a matter of fact, in what follows we will show that the action of the subalgebra $\gl(k+1)$ inside $\sym(2k+2)$ on an arbitrary element in 
$\mcH_{\mu}(\mR^{(k+1)m},\mC)$ generates a finite-dimensional $\gl(k+1)$-module inside the space of Howe harmonics. 
\\
Before doing so, we will first introduce a few notations. Recall that the general Lie algebra $\gl(k+1)$ is spanned by the basis elements $E_{ij}$ with 
$0 \leq i,j\leq k$. These elements $E_{ij}$ can then either be seen as matrices in $\mC^{(k+1) \times (k+1)}$, or as (skew) Euler operators in $(k+1)$ vector 
variables. The former relies on the identification $(E_{ij})_{kl} = \delta_{ik}\delta_{jl}$, the latter on $E_{ii} \mapsto \mE_i + \frac{m}{2}$ and 
$E_{ij} \mapsto \langle u_i, \partial_j \rangle$, for all $i,j = 0, \ldots, k$ and $i\neq j$ (recall that $x \equiv u_0$). The finite-dimensional irreducible 
representations for this Lie algebra are in one-to-one correspondance with $(k+1)$-tuples $\mu=(\mu_0,\mu_1, \ldots, \mu_k)\in \mC^{k+1}$ 
for which $\mu_i-\mu_{i+1} \in \mZ^+$, for all $0 \leq i \leq k$; this is then called the highest weight (HW) of the corresponding representation $\mV(\mu)$. 
The corresponding highest weight vector (HWV) is a unique element $v_\mu\in \mV(\mu)$ for which the relations 
$E_{ii}v_{\mu} = \left(\mu_i+\frac{m}{2}\right) v_{\mu}$ and $E_{ij}v_{\mu} = 0$ for $i<j$ hold. 
Now, observe that every element $H(u_0,u_1,\ldots, u_k)\in \mcH_{\mu}\brac{\mR^{(k+1)m},\mC}$ satisfies 
the conditions for a HWV, which means that every such element generates a $\gl(k+1)$-module under the action of the negative root vectors, 
which we will denote as
\begin{equation*}
\mV(\mu_0, \ldots, \mu_k)^\ast \cong \left(\mu_0 + \frac{m}{2}, \cdots, \mu_k + \frac{m}{2}\right).
\end{equation*}
The upper index $^\ast$ is a shorthand notation for the shift of the HW over half the dimension, and will frequently be used below. Our strategy for the remainder of this section is based on the following observations:
\begin{enumerate}[label=(\roman*)]
\item First of all, we note that the space $\ker_{\lambda_0}^A\mcD_\lambda$ can be seen as a subspace of the space of Howe harmonics. Indeed, it suffices to pick up the elements of the appropriate degree of homogeneity, and to intersect this space with the kernel of the operators $\langle u_i, \partial_j\rangle$ for which $0 < i < j \leq k$ (note that $x \equiv u_0$ is excluded here). This ensures that the corresponding Howe harmonics take values in the space $\mcH_\lambda(\mR^{km},\mC)$. 
\item The space of Howe harmonics is built from the modules $\mV(\mu)^\ast$ introduced above, where one has to sum over the degrees of homogeneity (see below). This means that studying the intersection 
\begin{equation*}
\mcH^h(\mR^{(k+1)m},\mC) \cap \ker\big(\langle u_i, \partial_j\rangle : 0 < i < j \leq k\big)
\end{equation*}
is reduced to studying the intersection of certain $\gl(k+1)$-representations with the kernel of these operators $\langle u_i, \partial_j\rangle$. Using the fact that these skew Euler operators can be seen as the positive root vectors of the subalgebra $\gl(k) \subset \gl(k+1)$ reduces our problem to a problem from representation theory, which has been studied in the setting of multiplicities of weight spaces and transvector algebras (see below). 
\end{enumerate}
First of all, we note the following: 
\begin{lemma}
Each element $E_{ij}$ of the algebra $\gl(k+1)$ acts as an endomorphism on the (total) space of Howe harmonics in $(k+1)$ variables.  
\end{lemma}

\begin{proof}
This can easily be checked using the definition of Howe harmonics. It amounts to saying that the commutator action of the algebra $\gl(k+1) \subset \sym(2k+2)$ on the space span$(\langle \partial_i, \partial_j\rangle : 0 \leq i, i \leq k)$ is well-defined.
\end{proof}
In other words, if $H(u_0,u_1,\ldots,u_k)$ is a Howe harmonic, then each $\mC$-valued polynomial of the form
\begin{equation} \label{form}
\left[\sum_{(p_{ij})}\bigg(\prod_{i,j}E_{ij}^{p_{ij}}\bigg)\right]H(u_0,u_1,\cdots,u_k), \qquad p_{ij} \in \mN 
\end{equation}
is still a Howe harmonic in $(k+1)$ vector variables. The factor between brackets denotes an arbitrary `word' (the letters do not necessarily commute here) in the (skew) Euler operators generating $\gl(k+1)$. This can also be formulated in the following way. 
\begin{lemma}
The elements of the universal enveloping algebra $\mcU\big(\gl(k+1)\big)$ preserve the space of Howe harmonics in $(k+1)$ vector variables. 
\end{lemma}

Moreover, and this is a crucial observation, no other elements in $\mcU(\sym(2k+2))$ have this property. To explain what this means, remember that the operators 
$\langle u_i, u_j \rangle, \langle u_i, \partial_j \rangle$ and $\langle \partial_i, \partial_j \rangle$ $(0 \leq i, j \leq k)$ generate a model for the Lie algebra 
$\sym(2k+2)$. When decomposing polynomial vector spaces in $(k+1)$ vector variables in terms of irreducible modules for the orthogonal group, one needs two pieces of 
information: highest weights, referring to {\em which} summands to include, and the so-called embedding factors, referring to {\em how} to include these 
summands. They correspond to products of elements in the algebra 
$\sym(2k+2)$, i.e. elements in the algebra $\mcU\big(\sym(2k+2)\big)$. It is known from the PBW-theorem that we can always rearrange these products 
according to a chosen ordering. We therefore select the following: 
\begin{enumerate}[label=(\roman*)]
	\item first, all combinations $\langle u_i,u_j \rangle$ involving vector variables only are listed
	\item then, all elements in $\gl(k+1)$ are listed (combinations of a vector variable and a derivative)
	\item finally, all combinations $\langle \partial_i,\partial_j \rangle$ involving pure differential operators only are listed. 
\end{enumerate}
It follows that the only elements in $\mcU(\sym(2k+2))$ which can be used as embedding factors are elements in $\mcU(\gl(k+1))$. Indeed: combinations involving type (iii) will always act trivially on the space of simplicial harmonics, whereas combinations involving type (i) will always belong to the Fischer complement of the space of Howe harmonics. The latter statement is based on the fact that 
\begin{equation*}
\mcP(\mR^{(k+1)m},\mC) = \mcH^h(\mR^{(k+1)m},\mC) \oplus \bigg( \sum_{i \leq j} \langle u_i, u_j \rangle \mcP(\mR^{(k+1)m},\mC) \bigg)
\end{equation*}
the sum between brackets obviously not being direct. This direct sum is orthogonal with respect to the Fischer inner product, see Section 3. Suppose that we now 
choose $(k+1)$ positive integers $\lambda_0 \geq \lambda_1 \geq \ldots \geq \lambda_k$, where $\lambda_j$ stands for the degree of homogeneity in the 
variable $u_j$. If we now want to decompose the vector space $\mcH^h_{\lambda_0,\lambda}(\mR^{(k+1)m},\mC)$ into irreducible representations for $\so(m)$, 
it suffices to select from each of the modules $\mV(\mu_0,\mu_1,\ldots,\mu_k)^*$ for $\gl(k+1)$ generated by the elements in 
$\mcH_{\mu_0,\ldots,\mu_k}(\mR^{(k+1)m},\mC)$ all the weight spaces having the correct degree of homogeneity .
\begin{example}
Despite the fact that the case $k = 1$ is rather trivial, it still is useful to illustrate the procedure described above. Suppose that we want to decompose the vector 
space $\mcH^h_{\lambda_0,\lambda_1}(\mR^{2m},\mC)$, $\lambda_0 \geq \lambda_1$. 
We then need to consider the $\gl(2, \mC)$-modules generated by the spaces $\mcH_{p,q}(\mR^{2m},\mC)$, with $p \geq q$. These are given by
\begin{equation*}
\mV(p,q)^* = \mcH_{p,q} \oplus \langle u_2,\partial_1\rangle \mcH_{p,q} \oplus \cdots \oplus \langle u_2,\partial_1\rangle^{p-q} \mcH_{p,q}
\end{equation*}
where it is easily verified that only a limited number of these modules will contribute to the space $\mcH^h_{\lambda_0,\lambda_1}$. Selecting the ones showing the correct degree of homogeneity, we thus indeed have that
\begin{equation*}
\mcH^h_{\lambda_0,\lambda_1} = \mcH_{\lambda_0,\lambda_1} \oplus \langle u_2,\partial_1\rangle \mcH_{\lambda_0+1,\lambda_1-1} \oplus \cdots \oplus \langle u_2,\partial_1\rangle^{\lambda_1} \mcH_{\lambda_0+\lambda_1,0}
\end{equation*}
\end{example}

In the general case, the procedure becomes more complicated since the weight spaces in arbitrary $\gl(k+1)$-modules occur with higher multiplicity, which means that also the decomposition for $\mcH^h_{\lambda_0,\lambda}(\mR^{(k+1)m},\mC)$ will no longer be multiplicity-free. \\
Another crucial difference between the case $k = 1$ (see the example above) and the cases $k > 1$ is the following: whereas each of the weight spaces
 $\langle u_2,\partial_1\rangle^{j} \mcH_{p,q} \subset \mV(p,q)^*$ is harmonic in $u_1$, and therefore occurs as a subspace of the space of type A solutions for 
an appropriate operator $\mcD_\lambda$ with $\lambda = (\lambda_1,0,\ldots,0)$, this is no longer true for $k > 1$. The problem is that the weight spaces of 
the module $\mV(\mu_0,\ldots,\mu_k)^*$ do not necessarily satisfy the equations defining $\mcH_\lambda(\mR^{km},\mC)$. Again switching to the 
general notations from above, we can formulate this algebraic problem as follows: \\ 
We first need to intersect each representation $\mV(\mu_0,\ldots,\mu_k)^*$ with the kernel of the operators $E_{ij}$ $(1 \leq i < j \leq k)$, where the index $0$ is 
again excluded. The desired polynomials should thus satisfy the condition to be a HWV for the algebra $\gl(k)$. To this end, we define the subspace 
$\mV(\mu)^+$ of $\mV(\mu)^* \equiv \mV(\mu_0,\ldots,\mu_k)^*$, containing all HWV of the subalgebra $\gl(k) \subset \gl({k+1})$:
\begin{equation*}
\mV(\mu)^+ = \{\eta \in \mV(\mu)^*: E_{ij} \eta = 0, \; 1 \leq i < j \leq k\}.
\end{equation*}
Moreover, we introduce a notation for the set of weight spaces in $\mV(\mu)^*$ realising a copy of the $\gl(k)$-module with highest weight 
$\nu^* = (\nu_1,\ldots,\nu_k)^* = \left(\nu_1+\frac{m}{2},\ldots,\nu_k+\frac{m}{2}\right)$. This means that for each of the elements in the previous set, a 
subscript $\nu$ is added referring to the $\gl(k)$-module for which it actually defines a HWV: 
\begin{equation*}
\mV(\mu)^+_\nu = \{\eta \in \mV(\mu)^+: E_{ii}\eta = \left(\nu_{i}+\frac{m}{2}\right) \eta, \; 1\leq i \leq k\}.
\end{equation*}
As $\mV(\mu)^*$ is generated by the operators $E_{ij}$ acting on elements $H\in \mcH_{\mu}$, each element $\eta \in \mV(\mu)^+_\nu$ is to be seen as a 
particular element of the form (\ref{form}), 
with $h(x,u_1,\ldots,u_k) \in \mcH_{\mu}$. Recall that the dimension of the spaces $\mV(\mu)^+_\nu$ either is 
$0$ or $1$, with
\begin{equation*}
\dim \left(\mV(\mu)^+_\nu \right) = 1 \Leftrightarrow \mu_{i-1}-\nu_i \in \mZ^+ \mbox{\ and\ } \nu_i - \mu_i \in \mZ^+, \; \mbox{for all\ } i = 1, \ldots, k
\end{equation*}
which is called the {\em betweenness condition}, as it can be represented graphically --at least for integer values of $\mu_i$ or integer values shifted over half the dimension-- by
\begin{equation*}
\mu_0 \geq \nu_1 \geq \mu_1 \geq \nu_2 \geq \mu _2 \geq \cdots \geq \mu_{k-1} \geq \nu_k \geq \mu_k.
\end{equation*}
We can gather these findings in the following result.
\begin{proposition} \label{decomp_Ms}
For each vector space $\mcH_{\mu}$, the only summands inside the representation $\mV(\mu)^*$ of $\gl({k+1})$ contributing to the space of type A solutions of the higher spin Laplace operator in $k$ dummy vector variables are of the form
\begin{equation*}
\rho_{d_1,\cdots,d_k}\mcH_{\mu}
\end{equation*}
where $\rho_{d_1,\cdots,d_k} \in \mcU\big(\gl({k+1})\big)$ is an embedding factor which is homogeneous of degree $(d_1,\cdots,d_k)$ in $(u_1,\cdots,u_k)$. Moreover, the integers $d_j$ satisfy the following conditions: 
\begin{equation*}
\mu_0 \geq \mu_1 + d_1 \geq \mu_1 \geq \mu_2 + d_2 \geq \cdots \geq \mu_{k-1} \geq \mu_k + d_k \geq \mu_k
\end{equation*}
or $0 \leq d_p \leq \mu_{p-1} - \mu_p$ (with $1 \leq p \leq k$). These conditions follow from the branching rules.
\end{proposition}

In the next section, an explicit form for these embedding factors $\rho_{d_1,\cdots,d_k}$ is obtained, using results on raising and lowering operators in 
transvector algebras. Note that these factors will be unique up to a constant, which follows from the fact that the branching from $\gl({k+1})$ to $\gl(k)$ is 
multiplicity-free. Let us now formulate the main conclusion of this section.
\begin{theorem}
As a module for the orthogonal group, the space $\ker_{\lambda_0}^A \mcD_{\lambda}$, with $\lambda_0 \geq \lambda_1$ decomposes into the following irreducible summands: 
\begin{equation*}
\ker_{\lambda_0}^A \mcD_{\lambda} = \bigoplus_{(d_1,\cdots,d_k)}\rho_{d_1,\ldots,d_k}\mcH_{\mu_0,\ldots,\mu_k}
\end{equation*}
where $(\lambda_0,\ldots,\lambda_k)$ is a dominant weight satisfying 
\begin{equation*}
\big(\mu_0,\mu_1,\ldots,\mu_k\big) = (\lambda_0 + \sum_{i=1}^k d_i,\lambda_1 - d_1, \ldots,\lambda_k - d_k)
\end{equation*}
with $\lambda_i - \lambda_{i+1} \geq d_i \geq 0$ for $1 \leq i \leq k-1$ and $0 \leq d_k \leq \lambda_k$. 
\end{theorem}
\begin{proof}
First, it follows from the branching rules for $\gl({k+1})$ to $\gl(k)$ that no embedding factor $\rho_{d_1,\cdots,d_k}$ can have a net effect of the form 
$(\pm p,\mp p)$ on the homogeneity degree in two variables $(u_i,u_j)$, with $i, j \geq 1$ and $p \in \mN$. Indeed: 
\begin{equation*}
(\mu_1,\ldots,\mu_k)^* \subset (\mu_0,\ldots,\mu_k)^*\bigg\vert^{\gl({k+1})}_{\gl(k)}
\end{equation*}
and any other summand which comes from the branching is obtained by adding positive integers $d_1, \ldots, d_k$ to resp. $\mu_1, \ldots, \mu_k$. 
This implies that the net effect of the factor $\rho_{d_1,\ldots,d_k}$ can always be represented as a leading term of the form 
\begin{equation}
\label{leadingterm}
\rho_{d_1,\ldots,d_k} = E_{21}^{d_1}\ldots E_{(k+1)1}^{d_k} + \ldots,
\end{equation}
where the numbers $(d_1,\ldots,d_k)$ satisfy the betweenness conditions coming from the branching. 
If we then fix the numbers $(\lambda_0,\ldots,\lambda_k)$, it suffices to find all the $(k+1)$-tuples $(\mu_0,\ldots,\mu_k)$ for which there exist 
positive integers $d_j$ such that we have an inclusion 
\begin{equation*}
\rho_{d_1,\ldots,d_k}\mcH_{\mu_0,\ldots,\mu_k} \subset \ker_{\lambda_0}^A \mcD_{\lambda}.
\end{equation*}
This is only possible if the conditions 
\begin{equation*}
\big(\mu_0 - \sum_{i=1}^k d_i,\mu_1 + d_1,\ldots,\mu_k + d_k\big) = (\lambda_0,\ldots,\lambda_k) 
\end{equation*}
on the degrees of homogeneity are satisfied, and if moreover
\begin{equation*}
\left\{
			\begin{array}{ccccc}
				\mu_0 - \mu_1 & \geq & d_1 & \geq & 0\\
				\mu_1 - \mu_2 & \geq & d_2 & \geq & 0\\
				& \vdots & \\
				\mu_{k-1} - \mu_k & \geq & d_k & \geq & 0.
			\end{array}
		\right. 
\end{equation*}
These are the conditions coming from the branching rules. Using the restrictions on the homogeneity, this can also be rewritten as 
$\lambda_i - \lambda_{i+1} \geq d_i \geq 0$, for all $1 \leq i < k$, and $\mu_k = \lambda_k - d_k$. This last expression tells us that 
$0 \leq d_k \leq \lambda_k$.
\end{proof}

\section{Relation with transvector algebras}

The aim of this section is to obtain explicit expressions for the embedding factors $\rho_{d_1,\cdots,d_k}$, i.e. the elements in $\mcU\big(\gl({k+1})\big)$
realising the decomposition of the space $\ker_{\lambda_0}^A \mcD_{\lambda}$ into irreducible summands under the orthogonal group.
Let us therefore take a look at the Mickelsson algebra $S(\gl(k+1), \gl(k))$, constructed as explained in \cite{Molev}. The generators of this algebra are given 
by the elements $z_{i0}$ and $z_{0i}$, $i = 1, \ldots, k$:
\begin{align*}
	z_{i0} &= \sum_{i>i_1>\cdots>i_s\geq 1}E_{ii_1}E_{i_1i_2}\cdots E_{i_{s-1}i_s}E_{i_s 0}(h_i-h_{j_1})\cdots (h_i-h_{j_r})\\
	z_{0i} &= \sum_{i<i_1<\cdots<i_s\leq k} E_{i_1 i}E_{i_2i_1}\cdots E_{i_{s}i_{s-1}}E_{0 i_s}(h_i-h_{j_1})\cdots (h_i-h_{j_r}).
\end{align*}
In these definitions, $s$ runs over nonnegative integers, $h_i := E_{ii}-i$ and $\{j_1, \ldots, j_r\}$ is the complementary subset to $\{i_1, \ldots, i_s\}$ in the set $\{1, \ldots, i-1\}$ or $\{i+1, \ldots, k\}$. For example, when $k = 3$ we have that
\begin{equation*}
z_{30} = E_{30}(h_3-h_1)(h_3-h_2) + E_{32}E_{20}(h_3-h_1) + E_{31}E_{10}(h_3-h_2) + E_{32}E_{21}E_{10}.
\end{equation*}
The properties of the extremal projector then lead to the following lemma, which was proved in \cite{Molev}.
\begin{lemma} \label{lemma4}
Let $\eta \in \mV(\mu)_\nu^+$, $\nu = (\nu_1,\ldots,\nu_k)$. Then, for any $i = 1, \ldots, k$, we have
\begin{equation*}
z_{i0} \eta \in \mV(\mu)_{\nu+\varepsilon_i}^+, \qquad z_{0i} \eta \in \mV(\mu)_{\nu-\varepsilon_i}^+,
\end{equation*} 
where the weight $\nu\pm \varepsilon_i$ is obtained from $\nu$ by replacing $\nu_{i}$ by $\nu_{i} \pm 1$.
\end{lemma}

In the present setting of solutions of higher spin operators, the lemma can be reformulated as follows: the operators $z_{i0}$ and $z_{0i}$, $i=1,\ldots, k$ will map a type A solution of a higher spin Laplace operator to another type A solution (be it for another operator, since the degree of homogeneity will change). More explicitly, the following results hold.
\begin{corollary}
For every polynomial $P(x;U) \in  \ker_{\lambda_0}^A \mcD_{\lambda_1, \ldots, \lambda_k}$, we have
	\begin{align*}
		& z_{i0}P(x;u_1,\cdots,u_k) \in \ker_{\lambda_0-1}^A \mcD_{\lambda_1, \cdots, \lambda_{i-1}, \lambda_{i}+1, \lambda_{i+1}, \cdots, \lambda_k} \\
		& z_{0i}P(x;u_1,\cdots,u_k) \in \ker_{\lambda_0+1}^A \mcD_{\lambda_1, \cdots, \lambda_{i-1}, \lambda_{i}-1, \lambda_{i+1}, \cdots, \lambda_k}.
	\end{align*}
\end{corollary}

\begin{example} 
When $k=2$, we have that
\begin{equation}\label{raisinglowering}
\begin{split}
z_{10} &= E_{10} = \langle u_1, \partial_{x} \rangle \\
z_{20} &= E_{21}E_{10}+E_{20}(h_2-h_1) = \langle u_2, \partial_1 \rangle\langle u_1, \partial_{x} \rangle + \langle u_2, \partial_{x} \rangle(\mE_2 - \mE_1 -1) \\
z_{01} &= E_{21}E_{02}+E_{01}(h_1-h_2) = \langle u_2, \partial_1 \rangle\langle x, \partial_2 \rangle + \langle x, \partial_1 \rangle(\mE_1 - \mE_2 +1) \\
z_{02} &= E_{02} = \langle x, \partial_2 \rangle.
\end{split}		
\end{equation}
\end{example}

In view of Lemma \ref{lemma4}, $z_{20}$ raises the degree in $u_2$ by one. Reconsidering the space $\ker_{3}^A \mcD_{1,1}$, we can now write its direct sum decomposition in terms of the explicit embedding factors:
\begin{equation*}
\ker_{3}^A \mcD_{1,1} = \mcH_{3,1,1} \oplus (\langle u_2, \partial_1 \rangle\langle u_1, \partial_x \rangle + \langle u_2, \partial_x \rangle(\mE_2 - \mE_1 -1))\mcH_{4,1}.
\end{equation*}
The Euler operators will only produce multiplicative constants, since they act on homogeneous polynomials. In this way, we also see the aforementioned leading terms in the example, up to a multiplicative constant.
\begin{lemma} \label{basis}
Let $\nu$ satisfy the betweenness condition stated above, and let $v_{\mu}$ be the highest weight vector of the module $\mV(\mu)$. Then the elements
\begin{equation*}
v_{\mu}(\nu) := z_{10}^{d_1}\cdots z_{k0}^{d_k} \, v_{\mu}
\end{equation*} 
are nonzero, provided that $(d_1,\cdots,d_k)$ satisfies all conditions of Theorem 1. Moreover, the space $\mV(\mu)^+$ is spanned by these elements 
$v_\mu(\nu)$.
\end{lemma}

\begin{example} 
As before, take $k=2$ and $\lambda = (4,3,1)^*$, and consider the module $\mV(4,3,1)^*$ generated by the space $\mcH_{4,3,1}$. Lemma \ref{basis} then states that consecutive actions of the operators $z_{10}$ and $z_{20}$ will produce a basis of the space $\mV(4,3,1)^* \cap \ker \langle u_1, \partial_2 \rangle$. More precisely, we obtain the following spaces, corresponding to the $6$ possible choices for $\mu$, and the respective spaces of higher spin solutions to which they contribute:
\begin{equation*}
		\begin{array}{l|l}
			\mcH_{4,3,1} & \ker_{4}^A \mcD_{3,1}\\
			z_{10}\mcH_{4,3,1} & \ker_{3}^A \mcD_{4,1}\\
			z_{20}\mcH_{4,3,1} & \ker_{3}^A \mcD_{3,2}  \\
			z_{10}z_{20}\mcH_{4,3,1} & \ker_{2}^A \mcD_{4,2} \\
			z_{20}^2\mcH_{4,3,1} & \ker_{2}^A \mcD_{3,3}  \\
			z_{10}z_{20}^2\mcH_{4,3,1} & \ker_{1}^A \mcD_{4,3}. \\
		\end{array}
\end{equation*}
Note however that this is {\em not} the decomposition of $\ker_{4}^A \mcD_{3,1}$. Indeed, using the correct embedding factors, we get that the latter is equal to 
\begin{equation*}
\ker_{4}^A \mcD_{3,1} = \mcH_{4,3,1} \oplus z_{10}\mcH_{5,2,1} \oplus z_{10}^2\mcH_{6,1,1} \oplus z_{20}\mcH_{5,3,0} \oplus z_{10}z_{20}\mcH_{6,2,0} \oplus z_{10}^2z_{20}\mcH_{7,1,0}.
\end{equation*} 
\end{example}

Recall that the embedding factor, as a whole, should behave as $E_{10}^a E_{20}^b$, with this term itself as a leading term. This might not be so obvious from the definitions and lemmas stated above. Note though that the operators $z_{i0}$ actually are defined up to a constant factor. We can also use the corresponding generators of the transvector algebra 
\begin{equation*}
Z(\gl(k+1), \gl(k)),
\end{equation*}
by using the field of fractions $R(\mathfrak{h})$. Hence, it is possible to divide $z_{i0}$ by the factor $(h_i-h_{i-1})\ldots(h_i-h_1)$, whence the resulting operators $s_{i0}$ (and likewise $s_{0i}$) take the form 
\begin{align*}
	s_{i0} &= \sum_{i>i_1>\cdots>i_s\geq 1}E_{ii_1}E_{i_1i_2}\cdots E_{i_{s-1}i_s}E_{i_s0}\frac{1}{(h_i-h_{i_1})\cdots (h_i-h_{i_s})} = p_{\gl(k)} E_{i0}\\
	s_{0i} &= \sum_{i<i_1<\cdots<i_s\leq k}E_{i_1 i}E_{i_2 i_1}\cdots E_{i_s i_{s-1}}E_{0 i_s}\frac{1}{(h_i-h_{i_1})\cdots (h_i-h_{i_s})} = p_{\gl(k)} E_{0i}
\end{align*}
or still $s_{i0} = E_{i0} + \mcO_{i0}$, which confirms \eqref{leadingterm}. It is now easily seen that powers of the 
operators $s_{i0}$ or $s_{0i}$ indeed behave as the leading terms predicted earlier. For instance, after rescaling, the four operators in (\ref{raisinglowering}) 
become
\begin{equation*} 
	\left\{
		\begin{array}{lcl}
			s_{10} &=& \langle u_1, \partial_x \rangle \\
			s_{20} &=& \displaystyle \langle u_2, \partial_x \rangle + \langle u_2, \partial_1 \rangle\langle u_1, \partial_x \rangle\frac{1}{\mE_2 - \mE_1 -1}\\
			s_{01} &=& \displaystyle \langle x, \partial_1 \rangle + \langle u_2, \partial_1 \rangle\langle x, \partial_2 \rangle\frac{1}{\mE_1 - \mE_2 +1}\\
			s_{02} &=& \langle x, \partial_2 \rangle.
		\end{array}
	\right.
\end{equation*}
So, the embedding factors defined in Proposition \ref{decomp_Ms}, are given by 
\begin{equation*}
\rho_{d_1,d_2, \cdots, d_k} = s_{10}^{d_1} \cdots s_{k0}^{d_k},
\end{equation*}
in accordance with Lemma \ref{basis}.


\appendix
\section{Action of the conformal group}
\label{Appendix_inversion}
The standard definition of an invariant operator $\mcD$ on a $G$-homogeneous bundle $E_V$ over $M=\quotient{G}{P}$ is that 
the operator $\mcD$ commutes with the induced action of $G$ on the space $\Gamma(E_V)$ of sections of $E_V$,
 see e.g. \cite[section 1.4]{Cap2} for more details. 
In this section, we will use some results from \cite{BSSVL1} to explain how to obtain the harmonic inversion $\mcJ_R$ from equation (\ref{Harmonic_inversion}). 
Therefore we will use the homogeneous space $\quotient{G}{P}$ where $G$ is the conformal group $\Pin(1,m+1)$. The group $G$ 
is the double cover of the orthogonal group $\operatorname{O}(1,m+1)$. There is a nice way to describe this group using the so-called Vahlen matrices, see 
e.g. \cite{Ahlfors, Vahlen}. In order to introduce them, we first have to introduce some basic definitions on Clifford algebras. 
The (real) universal Clifford algebra $\mR_m$ is the algebra generated by an orthonormal basis $\{e_1,\ldots,e_m\}$ for the vector space $\mR^m$ endowed with the Euclidean inner product $\inner{u,v} = \sum_j u_jv_j$ using the multiplication rules 
\begin{equation*}
e_ae_b + e_be_a = -2\langle e_a,e_b\rangle = -2\delta_{ab} 
\end{equation*}
with $1 \leq a, b \leq m$. The Clifford group is a subgroup of $\mR_m$ defined by means of
\begin{equation*}
\Gamma(m)=\set{\prod_{j=1}^k x_j : x_j \in \mR^m\setminus\set{0}, k\in \mN}.
\end{equation*}
In the following theorem, we also need the main involution (reversion) on $\mR_m$ which is defined on basis vectors as 
$e^{\ast}_{i_1 \ldots i_k}:=e_{i_k \ldots i_1}=(-1)^{\frac{k(k-1)}{2}}e_{i_1 \ldots i_k}$. Here the notation $e_{i_1 \ldots i_k}$ is used for the 
product $e_1e_2\ldots e_k$. We then have:
\begin{theorem}
Suppose $g\in \Pin(1,m+1)$ is a conformal transformation, then there exists a $2\times 2$-matrix
\begin{equation*}
A_g:=\begin{pmatrix}
a & b \\
c & d \\
\end{pmatrix}
\end{equation*}
such that: 
\begin{itemize}
\item[(i)] $a,b,c,d \in \Gamma(m)\cup \set{0}$.
\item[(ii)] $ab^{\ast},cd^{\ast},c^{\ast}a,d^{\ast}b\in \mR^m$.
\item[(iii)] $ad^{\ast}-bc^{\ast}=\pm 1$.
\end{itemize}
\end{theorem}

All conformal transformations can be written in the form $T(x)=(ax+b)(cx+d)^{-1}$, where $x\in \mR^m$. Such maps are well-defined 
on the conformal compactification of the Euclidean space $\mR^m$, which is the $m$-dimensional sphere $S^m$. The group 
$G$ acts transitively on $S^m$ and the isotropy group $P$ of the point $0\in \mR^m$ is the subgroup of $G$ where the 
associated Vahlen matrices have the property that $b=0$. The group $P$ is a parabolic subgroup of $G$. We now have a specific 
realisation of $S^m=\quotient{G}{P}$, which turns the $m$-dimensional sphere into a homogeneous space. 
\par
In what follows, we shall only consider vector bundles associated to irreducible $P$-representations. Due to condition $(iii)$, elements of $P$ have 
the property that $a^{-1}=\pm d^{\ast}$. It is known that the group $P$ is a semi-direct product of $G_0=\Pin(m)\times \mR^+$ and a 
commutative normal subgroup. This means that its irreducible representations are tensor products of irreducible representations of 
$\Pin(m)$ with a one-dimensional representation of $\mR^+$. They are classified by a highest weight $\lambda$ for $\Pin(m)$ and 
by a complex number $\nu\in \mC$, called the conformal weight. 
The normal subgroup of $P$ acts trivial on irreducible representations due the Engel's theorem. 
For $p\in P$, the entry $a\in \Gamma(m)$ of the associated Vahlen matrix 
\begin{equation*}
A_p:=\begin{pmatrix}
a & 0 \\
c & d \\
\end{pmatrix}
\end{equation*}
has nonzero spinor-norm, i.e. $\mcN(a)^2:=a\bar{a}\neq 0$. Here $\bar{a}$ is the Clifford conjugation, defined on basis vectors 
as $\bar{e}_{i_1 \ldots i_k}:=(-1)^ke_{i_k \ldots i_1}$.
This means that $a\in \Gamma(m)$ can be written as the product of $\frac{a}{\mcN(a)}\in \Pin(m)$ and $\mcN(a)\in \mR^+$. 
If $\rho_{\lambda}:\Pin(m)\longrightarrow \Aut\brac{\mV_{\lambda}}$ is an irreducible $\Pin(m)$-representation with highest 
weight $\lambda$ and $\nu \in \mC$, then we denote by $\rho_{\lambda}^{\nu}$ the irreducible representation of $P$ on 
$\mV_{\lambda}$, which is given by
\begin{equation*}
\rho_{\lambda}^{\nu}(p)\comm{v}:=\mcN(a)^{2\nu}\rho_{\lambda}\brac{\frac{a}{\mcN(a)}}\comm{v},
\end{equation*}
where $v\in \mV_{\lambda}$, $p\in P$ and $A_p=\begin{pmatrix} a & 0 \\ c & d \\ \end{pmatrix}$.
\par
In this paper, we only discuss differential operators acting on sections of 
homogeneous bundles over the open subset $\mR^m\subset S^m$, which can be considered embedded into the sphere 
$S^m$ by the map 
\begin{equation*}
i:\mR^m\longrightarrow S^m: x \mapsto \begin{pmatrix}1 & x \\ 0 & 1 \\ \end{pmatrix}P.
\end{equation*}
Such an embedding makes it possible to identify the space of sections $\Gamma(E_V)$ with the space of smooth functions 
$\mcC^{\infty}\brac{\mR^m,\mV_{\lambda}}$, i.e. for every section $s\in \Gamma(E_V)$ defined on the open subset 
$\mR^m\subset S^m$, we define the associated smooth function as follows:
\begin{equation*}
f:\mR^m \longrightarrow \mV_{\lambda}: x \mapsto f(x):=s(i(x)).
\end{equation*}
The induced action of $G$ on the space of smooth $\mV_{\lambda}$-valued functions is then given by (see e.g. \cite{BSSVL1})
\begin{equation}\label{G_action_functions}
g\cdot f(x):=\mcN(cx+d)^{2\nu}\rho_{\lambda}\brac{\frac{(cx+d)^{\ast}}{\mcN(cx+d)}}\Bigg[ f\brac{\frac{ax+b}{cx+d}}\Bigg],
\end{equation}
where $g^{-1}=\begin{pmatrix} a & b \\ c & d \\ \end{pmatrix}\in G$. Note that the absence of a minus sign in the exponent of $\mcN(cx+d)$ is due to a difference 
in conventions for the conformal weight.
\par
The Kelvin inversion on $\mR^m$ is given by $x\mapsto -\frac{x}{\norm{x}^2}$ which corresponds to the matrix 
\begin{equation*}
A_g:=\begin{pmatrix}
0 & 1 \\
-1 & 0 \\
\end{pmatrix}.
\end{equation*}
Using equation (\ref{G_action_functions}) and the fact that $\mcN(x)=x\bar{x}=\norm{x}^2$, this implies that 
the action of the Kelvin inversion on functions is given by
\begin{equation*}
A_g\cdot f(x)=\norm{x}^{2w} \rho_{\lambda}\brac{\frac{x}{\norm{x}}}\Bigg[ f\brac{\frac{x}{\norm{x}^2}}\Bigg].
\end{equation*}
If the function $f$ takes values in the space $\mcH_{\lambda}$, then this action becomes
\begin{equation}\label{Kelvin_g_action}
A_g\cdot f(x)=\norm{x}^{2-m}  f\brac{\frac{x}{\norm{x}^2},\frac{xu_1\bar{x}}{\norm{x}^2},\ldots,\frac{xu_k\bar{x}}{\norm{x}^2}}.
\end{equation}
Note that the conformal weight in case of the higher spin Laplace operator is given by $w=1-\frac{m}{2}$, which explains the exponent for $\norm{x}$. 
\begin{remark}
The expression from equation (\ref{Kelvin_g_action}) is different from the one from the one defined in \eqref{Harmonic_inversion}, but the difference is only up to an overall minus sign. 
\end{remark}

\end{document}